\documentclass{llncs}
\usepackage{graphicx}

\graphicspath{{./figs/}}
\usepackage{url}

\usepackage{extarrows}
\usepackage{graphics}
\usepackage{multirow}
\usepackage{amsfonts}
\usepackage{stmaryrd}

\usepackage{color}
\usepackage{subfigure}
\usepackage{amssymb}
\usepackage{hyperref}
\usepackage{algorithmic}
\usepackage{amsmath}
\usepackage[fleqn,tbtags]{mathtools}
\usepackage{txfonts}
\usepackage{float}

\usepackage{listings}

\usepackage{wasysym}
\usepackage{enumerate}

\makeatletter
\begingroup \catcode `|=0 \catcode `[= 1
\catcode`]=2 \catcode `\{=12 \catcode `\}=12
\catcode`\\=12 |gdef|@xcomment#1\end{comment}[|end[comment]]
|endgroup
\def\@comment{\let\do\@makeother \dospecials\catcode`\^^M=10\def\par{}}
\def\begincomment{\@comment\@xcomment}
\makeatother
\newenvironment{comment}{\begincomment}{}

\usepackage{url}
\usepackage{cite}

\newcommand{\techreport}[2]{#2}

\newcounter{NumofComms}
\newcommand\mynote[2]{\textcolor{red}{\sf #1: $\clubsuit$ #2$\clubsuit$}\addtocounter{NumofComms}{1}}

\def\SEQ#1{\mynote{SE\arabic{NumofComms}?}{#1}}

\def\dbrkts#1{{ [\! [ #1 ] \! ]}}

\def\LTS{\mathcal{S}}
\def\Reach{\mathit{Reach}}
\def\Cover{\mathit{Cover}}

\def\class{\mathit{class}}
\def\classes{\mathcal{C}}
\def\objects{O}
\def\allobjects{\mathbf{\mathcal{O}}}
\def\absnodes{\mathcal{V}}

\def\err{\mathit{Err}}
\def\idthis{\mathsf{this}}
\def\idnew{\mathsf{new}}
\def\idnull{\mathsf{null}}
\def\main{\mathsf{main}}
\def\retval{\mathsf{ret}}
\def\This{\mathit{this}}
\def\Null{\mathit{null}}
\def\nil{\mathit{nil}}
\def\curq{\mathit{q}}
\def\store{\mathit{st}}
\def\erroratr{\mathit{error}}

\def\error{\mathit{U}_{\mathit{err}}}
\def\abserror{\error^\#}

\def\havoc#1{\textrm{havoc}(#1)}
\def\assume#1{\textrm{assume}(#1)}

\def\threeabs{\mathit{cart}}

\def\heapabs{\mathit{heap}}
\def\concconfigs{\mathit{Conf}}
\def\conconfs{\concconfigs}
\def\absconfs{\conconfs^\#}
\def\absdom{D^\#}
\def\init{\mathit{U}_0}
\def\absinit{\init^\#}
\def\idealconf{I^\#}
\def\oidealconf{J^\#}
\def\absify#1{[#1]}
\def\allqidealconfs{\mathit{QIdlConf}^\#}
\def\allidealconfs{\mathit{IdlConf}^\#}
\def\allfinidealconfs{\allidealconfs_0}
\def\idealconfs{\mathcal{I}_\ideal}
\def\idealinit{\mathcal{I}_0}
\def\idealtrans{\ideal}
\def\idealwiden{\nabla_\ideal}
\def\extrapol{\chi_\ideal}
\def\transrel#1{\stackrel{#1}{\rightarrow}}

\def\idealtransrel#1{\mathrel{\transrel{#1}\!\!{}^\#_\ideal}}

\def\setconcconfigs{{\mathit{D}}}
\def\scope{{\mathit{scope}}}

\def\nlf{\mathit{nl}}
\def\twoheadsearrow{\ensuremath{\rotatebox[origin=c]{-40}{$\twoheadrightarrow$}}}

\newcommand{\copylayer}[1]{\stackrel{#1}{\twoheadrightarrow}}
\newcommand{\unfoldlayer}[1]{\mathrel{\twoheadsearrow_{#1}}}
\newcommand{\fold}{\mathit{fold}}
\newcommand{\reduce}{\mathit{reduce}}

\def\booltrue{1}
\def\boolfalse{0}

\def\var{x}
\def\Exp{{\mathsf{Exp}}}
\def\Pred{{\mathsf{Pred}}}

\def\Op{{\mathsf{Op}}}
\def\op{{\mathit{op}}}

\def\wlp{{\mathsf{wp}}}
\def\post{{\mathsf{post}}}
\def\abspost{{\post^\#}}
\def\reppost{{\mathsf{Post}^\#}}

\def\preset{{\mathsf{wp_s}}}
\def\postset{{\mathsf{sp_s}}}

\def\unarpreds{\mathit{AP}}
\def\funcpreds{\mathit{AR}}
\def\pv{\eta}
\def\absstore{\store}

\def\bool{\mathsf{h}}

\def\heap{\mathsf{h}}
\def\ideal{\mathsf{idl}}

\def\lsp{\mathit{lsp}}

\def\interface{\mathit{DPI}}
\def\objmap{\varphi}

\newcommand{\dc}{{\downarrow}}
\newcommand{\uc}{{\uparrow}}
\def\Idl{\mathit{Idl}}

\newcommand{\den}[1]{\llbracket #1 \rrbracket}
\newcommand{\set}[1]{\{#1\}}
\newcommand{\pset}[2]{\set{\,#1\mid#2\,}}
\def\powerset{\mathcal{P}}
\def\fin{\mathrm{fin}}
\def\pto{\rightharpoonup}
\def\lfp{\mathit{lfp}}
\newcommand{\layer}[2]{#1_{\geq #2}}
\newcommand{\ite}[3]{\mathrm{if}\; #1\; \mathrm{then}\; #2\; \mathrm{else}\; #3}

\newcommand{\BB}{\mathbb{B}}
\newcommand{\NN}{\mathbb{N}}

\newcommand{\myie}{i.e.,}

\newcommand{\srcf}[1]{\textsf{#1}}
\newcommand{\srct}[1]{\texttt{#1}}

\newtheorem{exam}{exam}
\newtheorem{examp}[exam]{Example}

\newsavebox{\mylistingbox}
\newsavebox{\mylistingboxx}

\newsavebox{\mylistingboxa}
\newsavebox{\mylistingboxb}
\newsavebox{\mylistingboxc}

\def\ao{\mathit{O}}
\def\ar{\store}

\newcommand{\roles}{\mathit{R}}
\newcommand{\objectlabel}{\mathit{AL}}

\def\labelao{{\mathit{l}}}
\def\nameao{{\mathit{n}}}

\def\nestingao{{\mathit{nl}}}
\def\classFromLabel{{\mathit{class}}}

\newcommand\pfun{\mathrel{\ooalign{\hfil$\mapstochar\mkern5mu$\hfil\cr$\to$\cr}}}

\def \sig{{\mathit{sig}}}

\def\method{{\mathcal{M}}}
\def\absheap{\mathcal{H}}
\def\errheap{\mathcal{E}}
\def\conheap{\mathcal{G}}
\def\maps{{\mathit{\Omega}}}

\def\preobj#1{[#1]}
\def\nd{nd}
\def\err{err}

\def\ncallee{\mathbf{callee}}
\def\narg#1{\mathbf{arg_{#1}}}
\def\nnew#1{\mathbf{new_{#1}}}
\def\nscope#1{\mathbf{scope_{#1}}}

\def\srca#1{\mathit{src(#1)}}

\begin{document}
\mainmatter              
\title{A Notion of Dynamic Interface for Depth-Bounded Object-Oriented Packages}
\author{
Shahram Esmaeilsabzali\inst{1}\thanks{Shahram Esmaeilsabzali was at MPI-SWS when this work was done.} \and Rupak Majumdar\inst{2} \and Thomas Wies\inst{3} \and 
Damien Zufferey\inst{4}\thanks{Damien Zufferey was at IST Austria when this work was done.}
}

\institute{$^1$University of Waterloo\quad\quad $^2$MPI-SWS\quad\quad $^3$NYU \quad\quad $^4$MIT CSAIL\\
\email{sesmaeil@uwaterloo.ca},\email{rupak@mpi-sws.org}, \email{wies@cs.nyu.edu}, \email{zufferey@csail.mit.edu}
}

\maketitle              

\begin{abstract}
Programmers using software components have to follow protocols that specify when it is legal to call particular methods with particular arguments.
For example, one cannot use an iterator over a set once the set has been changed directly or through another iterator.
We formalize the notion of \emph{dynamic package interfaces} (DPI), which generalize state-machine interfaces for single objects, and give an algorithm to statically compute a sound abstraction of a DPI.
States of a DPI represent (unbounded) sets of heap configurations and edges represent the effects of method calls on the heap.
We introduce a novel heap abstract domain based on depth-bounded systems to deal with potentially unboundedly many objects and the references among them.
We have implemented our algorithm and show that it is effective in computing representations of common patterns of package usage, such as relationships between viewer and label, container and iterator, and JDBC statements and cursors.

\end{abstract}
%
\begin{comment}
Programmers using a software component have to follow protocols that specify when it
is legal to call particular methods of objects with particular arguments.
For example, 
one cannot use an iterator over a set once the set has been changed directly
or through another iterator. 
We formalize the notion of {\em dynamic package interfaces} (DPI), which generalize
state-machine interfaces for single objects studied previously, 
and give an automatic static algorithm to compute a sound abstraction of a dynamic
package interface.
States of a DPI represent (unbounded) sets of heap configurations and edges
represent the effects of method calls on the heap.
%
Technically, we introduce a novel heap abstract domain
to deal with potentially unboundedly many object instances and their inter-relations,
extending shape analysis with ideal abstractions for depth-bounded systems.
We have implemented our algorithm for a Java-like source language, and show that
our algorithm is effective in computing representations of common patterns
of package usage, such as relationships between
viewer and label, container and iterator, and JDBC statements and cursors.
\end{comment}

\section{Introduction}

Modern object-oriented programming practice uses packages to encapsulate
components, allowing programmers to use these packages through well-defined
application programming interfaces (APIs).
While programming languages such as Java or C\#
provide a clear specification of the {\em static} APIs
of components in terms of classes and their (typed) methods,
there is usually no specification of the {\em dynamic} behavior of
packages that constrain the temporal ordering of method calls
on different objects.
For example, one should invoke the {\em lock} and {\em unlock}
methods of a lock object in alternation; any other sequence raises an exception.
More complex constraints connect method calls on objects of different classes.
For example, in the Java Database Connectivity (JDBC) package, a \srcf{ResultSet} object, 
which contains the result of a database query executed by a \srcf{Statement} object, 
should first be closed before its corresponding \srcf{Statement} object can execute a new query.

In practice, such temporal constraints are not formally specified, but explained
through informal documentation and examples, leaving programmers susceptible to 
bugs in the usage of APIs.
Being able to specify dynamic interfaces for components that capture these 
temporal constraints clarify constraints imposed by the package on client code. 
Moreover, program analysis tools may be able to automatically check 
whether the client code invokes the component correctly according to such an interface. 

Previous work on mining dynamic interfaces through static and dynamic techniques
has mostly focused on the single-object case (such as a lock object) 
\cite{Whaley02:Automatic,Alur05:Synthesis,Li05:PR,Henzinger05:Permissive,DBLP:conf/fase/GiannakopoulouP09}, and rarely on
more complex collaborations between several different classes (such as JDBC clients) interacting
through the heap \cite{Ramalingam02:Deriving,Nanda05:Deriving,Pradel12:Static}.
In this paper, we propose a systematic, static approach for extraction of 
dynamic interfaces from existing object-oriented code. 
Our work is closely related to the Canvas project \cite{Ramalingam02:Deriving}.
Our new formalization can express structures than could not be expressed in previous work (i.e. nesting of graphs).

More precisely, we work with {\em packages}, which are sets of classes.
A configuration of a package is a concrete heap containing objects from the package
as well as references among them. 
A \emph{dynamic package interface} (DPI) specifies, given a history of constructor and method calls on objects in the package,
and a new method call, if the method call can be executed by the package without causing an error.
In analogy with the single-object case, we are interested in representations of DPIs as finite state machines,
where states represent sets of heap configurations and transitions capture the effect of a method call
on a configuration. 
Then, a method call that can take the interface to a state containing erroneous configurations is not allowed by the interface,
but any other call sequence is allowed.

The first stumbling block in carrying out this analogy is that the number of states of an object, 
that is, the number of possible valuations of its attributes, as well as 
the number of objects living in the heap, can both be unbounded.
As in previous work \cite{Henzinger05:Permissive,Ramalingam02:Deriving}, we can bound the state space of a single 
object using {\em predicate abstraction},
that tracks the abstract state of the object defined by a set of logical formulas
over its attributes.
However, we must still consider unboundedly many objects on the heap and their inter-relationships.
Thus, in order to compute a dynamic interface, we must address the following challenges.
\begin{enumerate}
\item The first challenge is to define a finite representation for possibly unbounded heap configurations
and the effect of method calls.
For single-object interfaces, states represent a subset of finitely-many attribute valuations,
and transitions are labeled with method names.
For packages, we have to augment this representation for two reasons.
First, the number of objects can grow unboundedly, for example, through repeated calls
to constructors, and we need an abstraction to represent unbounded families of configurations.
Second, the effect of a method call
may be different depending on the receiver object and the arguments, 
and it may update not only the receiver and other objects transitively reachable from it,
but also other objects that can reach these objects.
\item The second challenge is to compute, in finite time, a dynamic interface using the preceding
representation.
For single-object interfaces \cite{Alur05:Synthesis,Henzinger05:Permissive}, interface construction roughly reduces to
abstract reachability analysis against the most general client (a program that non-deterministically calls all available methods in a loop).
For packages, it is not immediate that abstract reachability analysis will terminate, as our abstract
domains will be infinite, in general.
\end{enumerate} 
We address these challenges as follows. 
First, we describe a novel shape domain for finitely representing 
infinite sets of heap configurations as recursive unfoldings of nested graphs. 
Technically, our shape domain combines predicate abstraction~\cite{Sagiv02:Shape,Podelski05:Boolean}, 
for abstracting the internal state of objects, with 
sets of depth-bounded graphs represented as nested graphs~\cite{DBLP:conf/fossacs/WiesZH10}.
Each node of a nested graph is labelled with a valuation of the abstraction predicates 
that determine an equivalence class for objects of a certain class.

Second, we describe an algorithm to extract the DPI from this finite state abstraction
based on abstract reachability analysis of depth-bounded graph rewriting systems \cite{Zufferey12:Ideal}.
We use the insight that the finite state abstraction can be reinterpreted as a numerical program.
The analysis of this numerical program yields detailed information about how 
a method affects the state of objects when it is called on a concrete heap configuration, and how many objects are effected by the call.

We have implemented our algorithm on top of the Picasso abstract reachability tool for depth-bounded graph rewriting systems. 
We have applied our algorithm on a set of standard benchmarks written in a Java-like OO language, 
such as container-iterator, JDBC query interfaces, etc. 
In each case, we show that our algorithm produces an intuitive DPI for the package within a few seconds.
This DPI can be used by a model checking tool to check conformance of a client program using the package
to the dynamic protocol expected by the package.


%


\section{Overview: A Motivating Example}\label{sect:overview}

We illustrate our approach through a simple example.

\smallskip
\noindent
\emph{Example.}
Figure \ref{code:vl} shows two classes \srcf{Viewer} and \srcf{Label} in a package, 
adapted from \cite{Nanda05:Deriving}, and inspired by
an example from Eclipse's \srcf{ContentViewer} and \srcf{IBaseLabelProvider} classes.
A \srcf{Label} object throws an exception if its 
\srcf{run} or \srcf{dispose} method is called after the \srcf{dispose} method has been called on it.
There are different ways that this exception can be raised. 
For example, if a \srcf{Viewer} object sets its \srcf{f} reference to the same \srcf{Label} object twice, 
after the second call to \srcf{set}, the \srcf{Label} object, which is already disposed, raises an exception.
As another example, for two \srcf{Viewer} objects that have their \srcf{f} reference attributes point to the 
same \srcf{Label} object, when one of the objects calls its \srcf{done} method, 
if the other object calls its \srcf{done} method an exception will be raised.
An \emph{interface} for this package should provide possible configurations of the heap
when an arbitrary client uses the package, and describe all
usage scenarios of the public methods of the package that do not raise an exception.

\lstset{language=Java,basicstyle=\footnotesize\ttfamily,columns=spaceflexible}
\lstset{frame=none,linewidth = 7cm} \lstset{numberstyle=none}
\lstset{showstringspaces=false} \lstset{showspaces=false}
\lstset{morekeywords={pointcut,execution,target,Object,before,after}}
\lstset{escapechar=\%} \lstset{xleftmargin=4pt}

\begin{lrbox}{\mylistingboxa}%
\begin{small}
\begin{lstlisting}[mathescape]
class Viewer {
  Label f;
  public void Viewer() { 
    f := null; }
  public void run() { 
    if (f != null) f.run(); }
  public void done() { 
    if (f != null) f.dispose(); }
  public void set(Label l){
    if (f != null) f.dispose(); 
    f := l; }
}
\end{lstlisting}	
\end{small}
\end{lrbox}

\begin{lrbox}{\mylistingboxb}%
\begin{small}
\begin{lstlisting}[mathescape]
class Label {
  boolean disposed;
  public void Label() { 
    disposed := false; 
  }
  protected void run() { 
    if (disposed) throw new Exception(); }
  protected void dispose() { 
    if (disposed) throw new Exception(); 
    disposed := true; } 
}
\end{lstlisting}
\end{small}
\end{lrbox}
\newsavebox{\tempbox}

\begin{figure*}[t]
{\small
\subfigure[The \srcf{Viewer} class]{
\usebox{\mylistingboxa}
\label{fig:listitera}
}
\subfigure[The \srcf{Label} class]{
\usebox{\mylistingboxb}
\label{fig:listiterb}
}\\

\subfigure[Abstract heap $\mathit{H_0}$]{\raisebox{13mm}{\input{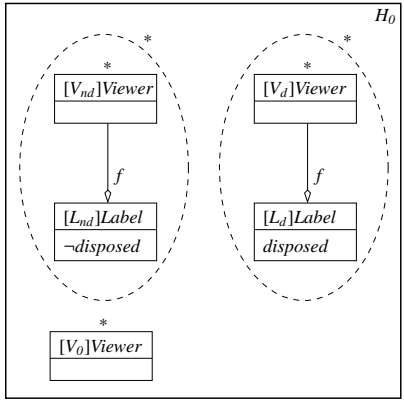}}
\label{fig:vlcover1}
}
\hspace{2em}
\subfigure[Abstract heap $\mathit{H_{Err}}$]{\input{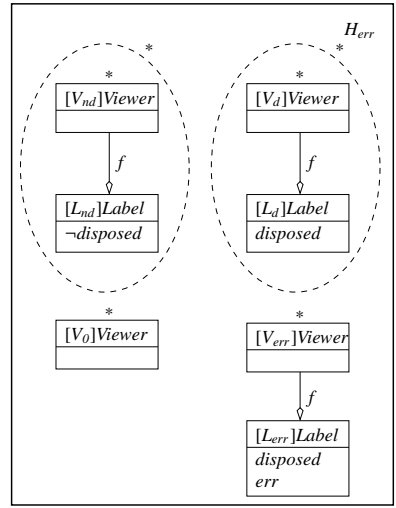}
\label{fig:vlcovererr}
}
}
\caption{A package consisting of \srcf{Viewer} and \srcf{Label} classes and its two abstract heaps}
\label{code:vl}
\end{figure*}

\smallskip
\noindent\emph{Dynamic Package Interface.}
Intuitively, an interface for a package summarizes all possible ways for a client to make calls into the package
(i.e., create instances of classes in the package and call their public methods).
In the case of single-objects, where all attributes are scalar-valued,
interfaces are represented as finite-state machines with transitions labeled
with method calls \cite{Whaley02:Automatic,Alur05:Synthesis,Henzinger05:Permissive}.
Each state $s$ of the machine represents a set $\dbrkts{s}$ of states of the object,
where a state is a valuation to all the attributes.
(In case there are infinitely many states, the methods of \cite{Alur05:Synthesis,Henzinger05:Permissive}
abstract the object relative to a finite set of predicates, so that the number of states is finite.)
An edge $s\smash{\xrightarrow{m}} t$ indicates that calling the method $m()$ from any state in $\dbrkts{s}$
takes the object to a state in $\dbrkts{t}$.
Some states of the machine are marked as errors: these represent inconsistent states, and method
calls leading to error states are disallowed.

Below, we generalize such state machines to packages.

\smallskip
\noindent\emph{States: Ideals over Shapes.}
The first challenge is that the notion of a state is more complex now.
First, there are arbitrarily many states: for each $n$, we can have a state with $n$ instances of \srcf{Label}
(e.g., when a client allocates $n$ objects of class \srcf{Label}); moreover, we can have more complex configurations
where there are arbitrarily many viewers, each referring to a single \srcf{Label}, where the \srcf{Label} may
have \srcf{disposed = true} or not. 
We call sets of (potentially unbounded) heap configurations \emph{abstract heaps}.

Our first contribution is a novel finite representation for abstract heaps.
We represent abstract heaps using a combination of \emph{parametric shape analysis} \cite{Sagiv02:Shape} 
and \emph{ideal abstractions for depth-bounded systems} \cite{Zufferey12:Ideal}.
As in shape analysis, we fix a set of unary predicates, and abstract each object w.r.t.\ these predicates.
For example, we track the predicate $\mathtt{disposed(}\mathit{l}\mathtt{)}$ to check if an object $l$ of type \srcf{Label} has
\srcf{disposed} set to true.
Additionally, we track references between objects by representing the heap as a nested graph whose nodes represent predicate abstractions of
objects and whose edges represent references from one object to another. Unlike in parametric shape analysis, references are always determinate and the abstract domain is therefore still infinite.

Figure~\ref{fig:vlcover1} shows an abstract heap $H_0$ for our example.
There are five nodes in the abstract heap.
Each node is labelled with the name of its corresponding class and a valuation of predicates, and represents an object of the specified class whose state satisfies the predicates.
Some nodes have an identifier in square brackets in order to easily refer to them.
For instance, $V_{\nd}$ represents a \srcf{Viewer} object and $L_d$ represents a \srcf{Label} object for which \srcf{disposed} is true.
Edges between nodes show field references: the edge between the $V_d$ and $L_d$ objects that is labeled with $f$ shows that objects
of type $V_d$ have an $f$ field referring to some object of type $L_d$.
Finally, nodes and subgraphs can be marked with a ``*''.
Intuitively, the ``*'' indicates an arbitrary number of copies of the pattern within the scope of the ``*''.
For example, since $V_d$ is starred, it represents arbitrarily many (including zero) \srcf{Viewer} objects
sharing a \srcf{Label} object of type $L_d$. 
Similarly, since the subgraph over nodes $V_d$ and $L_d$ is starred, it represents configurations with
arbitrarily many \srcf{Label} objects, each with (since $V_d$ is starred as well)
arbitrarily many viewers associated with it.

Figure~\ref{fig:vlcovererr} shows a second abstract heap $H_{\err}$. This one has two extra nodes in addition
to the nodes in $H_0$, and represents erroneous configurations in which the \srcf{Label} object is about to throw
an exception in one of its methods. (We set a special error-bit whenever an exception is raised, and
the node with object type $L_{\err}$ represents an object where that bit is set.)

Technically (see Sections~\ref{sec:abstractsem} and~\ref{sec:analysis}), nested graphs represent ideals of
downward-closed sets (relative to graph embedding) of configurations of depth-bounded predicate abstractions of the heap.
While the abstract state space is infinite, it is well-structured, and abstract reachability analysis can be done \cite{Abdulla96:General,Meyer08OnBoundednessInDepth,DBLP:conf/fossacs/WiesZH10}.

\smallskip
\noindent\emph{Transitions: Object Mappings.}
Suppose we get a finite set $\mathcal{S}$ of abstract heaps represented as above.
The second challenge is that method calls may have parameters and may change the state of the receiver object
as well as objects reachable from it or even objects that can reach the receiver.
As an example, consider a set container object with some iterators pointing to it.
Removing an element through an iterator can change the state of the iterator (it may reach the end), 
the set (it can become empty),
as well as other iterators associated with the set (they become invalidated and may not be used to traverse the set).
Thus, transitions cannot simply be labeled with method names, but must also indicate which abstract objects participate in the call
as well as the effect of the call on the abstract objects.
The interface must describe the effect of the heap in all cases, and all methods.
In our example, we can enumerate 14 possible transitions from $H_0$. 
To complete the description of an interface, we have to (1) show how a method call transforms the abstract heap,
and (2) ensure that each possible method call from each abstract heap in $\mathcal{S}$ ends up in an abstract heap
also in $\mathcal{S}$.

\begin{figure}[t]
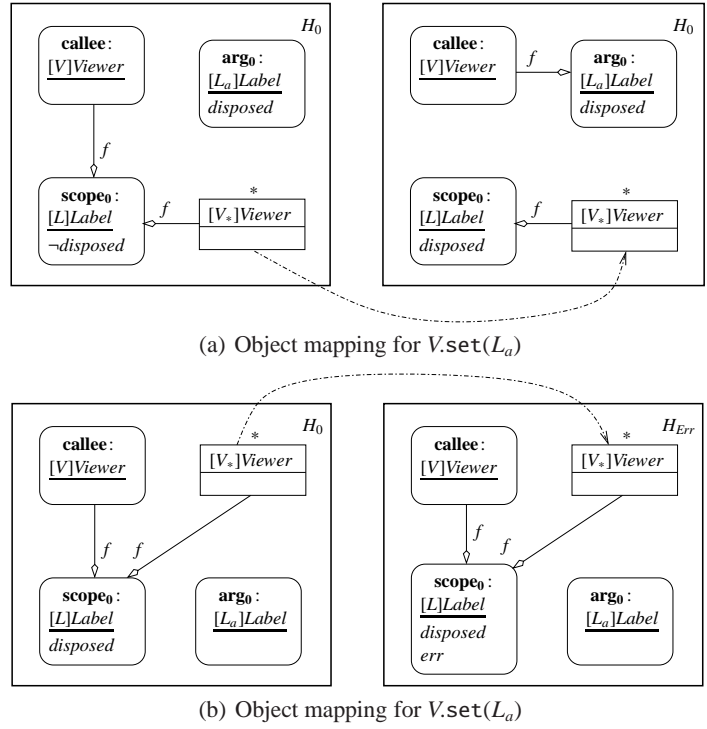
%
\centering
\subfigure[Object mapping for $V.\mathtt{set(}L_a\mathtt{)}$]{\input{figs/vl13.pstex_t}
\label{fig:vlt13}
}\qquad 
\subfigure[Object mapping for $V.\mathtt{set(}L_a\mathtt{)}$]{\input{figs/vl3.pstex_t}
\label{fig:vlt3}
}
\caption{Two object mappings for the package in Figure \ref{code:vl}}
\label{fig:vlmaps}
\end{figure}

Consider invoking the \srcf{set} method of a viewer in the abstract heap $H_0$.
There are several choices: one can choose in Figure~\ref{fig:vlcover1} an object of type $V_d$, $V_{nd}$, or $V_0$ as the callee, and pass it an object of type $L_d$ or $L_{nd}$.
Note that the method call captures the scenario in which one \emph{representative} object is chosen from each node
and the method is executed. Recall that, because of stars, a single node may represent multiple objects.
Figure~\ref{fig:vlt13} shows how the abstract heap is transformed if we choose a viewer pointing to a label which is not disposed as the callee and pass it a disposed label as argument.
The box on the left specifies the \emph{source} heap before the method call and the box on the right specifies the \emph{destination} heap after the method call.
A representative object in a method call is graphically shown by a rounded box and has a \emph{role name} that prefixes its object type. 
The source heap includes three representative objects with role names: $\ncallee$, $\narg{0}$, and $\nscope{0}$.
The $\ncallee$ and $\narg{0}$ role names determine the callee object and the parameter object of the method call, respectively.
The $\nscope{0}$ is a $\mathit{Label}$ object that is in the \emph{scope} of the method call: i.e., the method call affects its type or the valuation of its predicates.
Lastly, there is a fourth object in the left box that is not a single representative, but a starred object $V_*$ that represents all viewers other than the callee object that reference the object with role $\nscope{0}$.
The following properties hold.
First, both the source and the destination of the transition are $H_0$, hence,
the method call transforms objects in the abstract heap $H_0$ back to $H_0$.
Second, any object in $H_0$ that is not mentioned in the source box is untouched by the method call.
Third, each object in the left box is mapped to another representative object in the right box: The representative objects can be traced via their role names while the other objects via the arrows that specifies their new types (to model non-determinism, such an arrow can be a multi-destination arrow).
Thus, $V.\mathtt{set(}L_a\mathtt{)}$ transforms the callee object by changing its reference $f$ to the $L_a$ object that was the parameter of the method call.
The object $L$ that the callee referenced before the method call get the value of its disposed predicate changed to true after the method call.
All other objects represented by $V_*$ that reference $L$ continue referencing that object.

\begin{comment}
Third, updates to nodes are shown using directed hyperedges (dotted multi-destination lines in the figure).
The hyperedges are optionally labeled with ``1'' (e.g., the hyperedge from $L_{nd}$ to the destination $L_d$).
The hyperedge labeled 1 indicates that one object from the source node transfers to the state in the destination;
every other node in the source node moves (non-deterministically) to some destination of the hyperedge (that is not marked by ``1'').

So, the method call $V_{nd}.\mathtt{set}(L_d)$ moves one viewer object from $V_{nd}$ to $V_d$ (the callee), and all other
viewer objects in state $V_{nd}$ non-deterministically to either $V_d$ or $V_{nd}$.
At the same time, one $L_{nd}$ object (the old viewer pointed to by the callee) moves to $L_d$, and all other remain in $L_{nd}$. Finally, all $L_d$ objects (including the parameter of the call) remain in $L_d$.
Intuitively, this captures the situation that the previous (non-disposed) label object is now disposed,
and the viewer now points to a disposed label.
\end{comment}

The second transition, in Figure~\ref{fig:vlt3}, shows what happens if \srcf{set} is called on $V_d$ with any label.
This time, an error occurs, since the method call tries to dispose an already disposed label.
This is indicated by a transformation to the error node $H_{\err}$, and thus, is not allowed in the interface.

\smallskip
\noindent
\emph{Algorithm for Interface Computation.}
Our second contribution is an algorithm and a tool for computing the dynamic package interfaces in form of a state machine, as described above.
Conceptually, the DPI of a package is computed in two steps: (i) computing the \emph{covering set} of the package, which includes all possible configurations of the package, in a finite form; and (ii) computing the object mappings of the package using the covering set. 

\paragraph{Computing the Covering Set.} 
We introduce three layers of abstraction to obtain an overapproximation of the covering set of a package in a finite form. First, using a fixed set of predicates over the attributes of classes, we introduce a predicate abstraction layer.
Second, we remove from this predicate abstraction those reference attributes of classes that can create a chain of objects with an unbounded length; these essentially correspond to recursive data structures, such as linked lists.
We call these two abstraction layers the \emph{depth-bounded abstraction}. 
The soundness of depth-bounded abstraction follows soundness arguments similar to the ones for classic abstract interpretation. However, unlike the classic abstract interpretation of non-object--oriented programs, the depth-bounded abstraction of object-oriented packages does not in general result in a finite representation; e.g., we may still have an unbounded number of iterator and set objects, with each iterator object being connected to exactly one set object.
\begin{comment}
A key property of the depth-bounded abstraction domain, however, is that its elements can be ordered via a well-quasi order that is essentially based on the subset inclusion of objects.
Furthermore, the transition system of a package over a depth-bounded abstract domain results in a well-structured transition system.
\end{comment}

Our third abstraction layer, namely, \emph{ideal abstraction}, ensures a finite representation of the covering set of a package.
The domain of ideal abstraction is essentially the same as the domain of nested graphs.
The key property of this abstraction layer is that it can represent an unbounded number of depth-bounded objects as the union of a finite set of ideals, each of which itself is represented finitely.
The soundness of this abstraction layer follows from the general soundness result for the ideal abstraction of depth-bounded systems \cite{Zufferey12:Ideal}.

To compute the covering set of a package, we use a notion of \emph{most general client}. 
Intuitively, the most general client \cite{Henzinger05:Permissive} runs in an infinite loop; in each iteration of the loop,
it non-deterministically either allocates a new object, or picks an already allocated object,
a public method of the object, a sequence of arguments to the method, and invokes the method call on the object.
Using a widening operator over the sequence of the steps of the most general client, our algorithm is able to determine when the nesting level of an object needs to be incremented. 
Our algorithm terminates due to the fact that the ideal abstraction is a well-structured transition system.

\paragraph{Computing the Object Mappings.} 
The object mappings are computed using the covering set as starting point.
To compute the object mappings we let the most general client run one more time using the covering set as starting state of the system.
During that run we record what effect the transitions have.
For a particular transition we record, among other information, what are the starting and ending abstract heaps and the corresponding \emph{unfolded}, representative objects.
The nodes of the unfolded heap configurations are tagged with their respective roles in the transition.
Finally, we record how the objects are modified and extract the mapping of the object mapping.

\begin{comment}
Our second contribution is an algorithm and a tool for computing the dynamic package interfaces in form of a state machine, as described above.
To compute the interface of a package, we fix a set of predicates and then perform an abstract reachability analysis of the package
together with a {\em most general client}.
Intuitively, the most general client \cite{Henzinger05:Permissive} runs in an infinite loop; in each iteration of the loop,
it non-deterministically either allocates a new object of some class, or picks an already allocated object,
a public method of the object, a sequence of arguments to the method, and invokes the method call on the object.
This way, it explores all possible sequences of constructors and method calls.
The properties of ideal abstraction of depth-bounded systems \cite{Zufferey12:Ideal} ensures that the 
abstract reachability analysis terminates and produces a finite coverability tree for the package.
The nodes of the coverability tree represent abstract heaps and edges represent object maps on transitions.
We take the maximal nodes of this tree (relative to embedding of abstract heaps) as the states of our interface,
and the associated transitions as the transitions.
\end{comment}

In our example, there are two maximal nodes: $H_0$ and $H_{\err}$, where $H_{\err}$ denotes the error configurations. $H_0$ and $H_{\err}$ together represent the covering set of the package.
Accordingly, the interface shows that $H_0$ captures the ``most general'' abstract heap in the use
of this package; each ``correct'' method call corresponds to an object mapping over $H_0$. 
We omit showing the remaining 12 object mappings of the interface.

\section{Concrete Semantics}\label{sect:lang}

We now present a core OO language. 

\smallskip
\noindent
\emph{Syntax.}
%
For a set of symbols $X$ (including variables), we denote  by $\Exp.X$ and $\Pred.X$ the set of expressions 
and predicates respectively, 
constructed with symbols drawn from $X$.
We assume there are two special variables $\idthis$ and $\idnull$.  

In our language, a package consists of a collection of class definitions.
A class definition consists of a class name, a constructor method, a set of fields,
and a set of method declarations partitioned into public and protected methods.
A constructor method has the same name as the class, a list of typed arguments, and a body.
We assume fields are typed with either a finite scalar type (e.g., Boolean),
or a class name.
The former are called \emph{scalar} fields and the latter \emph{reference} fields.
Intuitively, reference fields refer to other objects on the heap.
Methods consist of a signature and a body.
The signature of a method is a typed list of its arguments and its return value.
The body of a method is given by a control flow automaton over the fields of the class. 
Intuitively, any client can invoke public methods, but only
other classes in the package can invoke protected ones.

A control flow automaton (CFA) over a set of variables $X$ and a set of operations $\Op.X$
is a tuple $F=(X,Q, q_0, q_f, T)$,
where $Q$ is a finite set of \emph{control states}, $q_0\in Q$ (resp.\ $q_f\in Q$)
is a designated initial state (resp.\ final state), and 
$T \subseteq Q\times \Op.X \times Q$ is a set of edges labeled with operations.

For our language, we define the set $\Op.X$ of \emph{operations} over $X$ to consist of: 
(i) \emph{assignments} $\idthis.x := e$, where $x\in X$ and $e \in \Exp.X$;
(ii) \emph{assumptions}, $\assume{p}$, where $p \in \Pred.(\{\idthis\} \cup X)$,
(iii) \emph{construction} $\idthis.x = \idnew(C(\bar{a}))$, where $C$ is a class name and
$\bar{a}$ is a sequence in $\Exp.X$,
and
(iv) \emph{method calls} $\idthis.x := \idthis.y.m(\bar{a})$, where $x, y\in X$. 

Formally, a class $C = (A, c, M_p, M_t)$, where $A$ is 
the set of fields, 
$c$ is the constructor,
$M_p$ is the set of public methods,
and
$M_t$ is the set of protected methods.
We use $C$ also for the name of the class.
A package $P$ is a set of classes.

We make the following assumptions.
First, all field and method names are disjoint.
Second, each class has an attribute $\retval$ used to return values from a method
to its callers.
Third, all CFAs are over disjoint control locations.
Fourth, a package is well-typed, in that assignments are type-compatible,
called methods exist and are called with the right number and types of arguments, etc.
Finally, it is not clear how the pushdown system and depth-bounded system mix and whether there exists an bqo that may accomodate both.
Therefore, we omit recursive method calls from our the analysis.

A {\em client} $I$ of a package $P$ is a class 
with exactly one method $\main$, such that
(i) for each $x\in I.A$, we have the type of $x$ is either a scalar or a class name from $P$,
(ii) in all method calls $\idthis.x = \idthis.y.m(\bar{a})$, $m$ is a public method of its class,
and 
(iii) edges of $\main$ can have the additional \emph{non-deterministic assignment} $\havoc{\idthis.x}$.
An OO program is a pair $(P, I)$ of a package $P$ and a client $I$.

\smallskip
\noindent
\emph{Concrete Semantics.}\label{sect:concsem}
We give the semantics of an OO program as a labeled transition system.
A \emph{transition system} 
$\LTS = (X,X_0,\rightarrow)$ consists of a set $X$ of states, 
a set $X_0 \subseteq X$ of initial states, and a transition relation
$\rightarrow \; \subseteq X \times X$. 
We write $x\rightarrow x'$ for $(x,x')\in\rightarrow$.

Fix an OO program $S = (P, I)$. 
It induces a transition system 
$(\concconfigs, \init, \rightarrow)$, 
with configurations $\concconfigs$, initial configurations $\init$, and transition relation $\rightarrow$ as follows.

Let $\allobjects$ be a countably infinite set of \emph{object identifiers} (or simply objects) and 
let $\class: \objects \to P \cup\set{I, \nil}$ be a function mapping each object identifier to its class.  
A \emph{configuration} $u \in \concconfigs$ is a tuple $(\objects,\This,\curq,\nu, \mathit{st})$, 
where $\objects \subseteq \allobjects$ is a finite set of currently allocated \emph{objects}, 
$\This \in \objects$ is the \emph{current object} (i.e., the receiver of the call to the method currently executed), 
$\curq$ is the \emph{current control state}, which 
specifies the control state of the CFA at which the next operation will be performed, 
$\nu$ is a sequence of triples of object, variable, and control location (the program stack),
and $\mathit{st}$ is a \emph{store}, which maps an object and a field to a value in its domain. 
We require that $\objects$ contains a unique \emph{null} object $\Null$ 
with $\mathit{class(null)} = \nil$. 
We denote by $\concconfigs$ the set of all configurations of $S$.

The set of \emph{initial configurations} $\init \subseteq \concconfigs$ is the set of configurations 
$u_0 = (\set{\Null, o_I},\This,\main.q_0,\varepsilon, \mathit{st})$ such that
(i) $\class(o_I) = I$, 
(ii) the current object $\This = o_I$, 
(iii) the value of all reference fields of all objects in the store is $\mathit{null}$ and all scalar
fields take some default value in their domain, 
and 
(iv) the control state is the initial state of the CFA of the main method of $I$ and the stack is empty. 

Given a store, 
we write $\store(e)$ and $\store(p)$ for the value of an expression $e$ or predicate $p$ evaluated in the store $\store$,
computed the usual way.
%

The transitions in $\rightarrow$ are as follows. 
A configuration $(\objects,\This,\curq,\nu, \store)$ moves to configuration $(\objects',\This',\curq',\nu', \store')$
if there is an edge $(\curq, \op, \curq')$ in the CFA of $\curq$ such that
\begin{itemize}

\item $\op = \idthis.x:=e$ and $\objects'  = \objects$, $\This' = \This$,
$\nu' = \nu$,
and $\store' = \store[(\This,x) \mapsto \store(e)]$.
\item $\op = \assume{p}$ and $\objects' = \objects$, $\This' = \This$,
$\nu' = \nu$,
$\store(p) = \booltrue$,  and $\store' = \store$.
\item $\op = \idthis.x:=\idthis.y.m(\bar{a})$ and 
$\objects'  = \objects$, $\This' = \This$,
$\nu' = (\This, x, \curq') \nu$, and $\curq' = m.q_0$, and the formal arguments of $m$ are
assigned values $\store(\bar{a})$ in the store.
\item $\op = \idthis.x :=\idnew(C(\bar{a}))$ and 
$\objects' = \objects \uplus \set{o}$ for a new object $o$ with $\class(o) = C$,
$\This' = o$,
$\nu' = (\This, x, \curq') \nu$, and $\curq' = c.q_0$ for the constructor $c$ of $C$, 
and the formal arguments of $c$ are assigned values $\store(\bar{a})$ in the store.

\item $\op = \havoc{\idthis.x}$: $\objects' = \objects$, $\This' =   \This$, and 
$\store' = \store[(\This,x) \mapsto v]$,   where $v$ is some value chosen non-deterministically from the domain of $x$. 
\end{itemize}
Finally, if $\curq$ is the final node of a CFA and $\nu = (o, x, q) \nu'$, and
the configuration $(\objects,\This,\curq,\nu, \store)$ moves to $(\objects,o,q,\nu', \store')$,
where $\store' = \store[o.x \mapsto \store(\This.\retval)]$.
If none of the rules apply, the program terminates.


To model error situations, we assume that each class has a field
$\err$ which is initially $0$ and set to $1$ whenever an error is encountered (e.g.,
an assertion is violated). 
An error configuration is a configuration $u$ in which there exists an object $o\in u.\objects$
such that $o.\err = 1$.
An OO program is \emph{safe} if it does not reach any error configuration. 

\begin{examp}
Figure~\ref{fig:iter_conc} depicts two configurations for a set of objects belonging to a ``set and iterator'' package. 
For the sake of brevity, we do not show the code for this package, but the functionality of the package is standard. 
The package has three classes, namely, \srct{Set}, \srct{Iterator}, and \srct{Elem}.
The \srct{Elem} class can create a linked list to store the elements of a \srct{Set} object.
An \srct{Iterator} object is used to traverse the elements of its corresponding \srct{Set} object via its \srct{pos} attribute as an index. 
It can also remove an element of the \srct{Set} object through its \srct{remove} method. 
An \srct{Iterator} object can perform these operations only if it has the same \emph{version} as its corresponding \srct{Set} object. 
The \srct{Iterator} version is stored in the \srct{iver} field and the \srct{Set} version in \srct{sver}.
In this example, we focus on the \srct{remove} method. 
The \srct{remove} method of an \srct{Iterator} object invokes the \srct{delete} method of its corresponding \srct{Set} object, passing its \srct{pos} attribute as a parameter.
The \srct{delete} method, in turn, deletes the \srct{pos}th \srct{Elem} object that is accessible through its \srct{head} attribute.
The version attributes of both the \srct{Iterator} and \srct{Set} objects are incremented, while the version attributes of other \srct{Iterator} objects remain the same.
The two configurations in Figure~\ref{fig:iter_conc} are abbreviated to show only the information relevant to this example.

The configuration 
\begin{equation*}
u  =  \mathit{(\{s,i_1,i_2.e_1,e_2\}, s,.,\langle (i_2,.,.) \rangle}, \mathit{\{((i1,iver),2), ((i2, iver), 2), \cdots \})}, 
\end{equation*}
\noindent depicted in Figure~\ref{fig:iter_conc}(a), is one of the configurations during the execution of $i_2.\mathtt{remove}$, namely the configuration immediately after executing $\srcf{this}\mathtt{.iter\_of.delete(}\srcf{this}\mathtt{.pos)}$.
After a number of steps, the computation reaches configuration 
\begin{equation*}
u' = \mathit{(\{s,i_1,i_2.e_1,e_2\}, s,.,\varepsilon,\{((i1, iver), 2), ((i2, iver), 3), \cdots \})}, 
\end{equation*}
depicted in Figure~\ref{fig:iter_conc}(b), which is the configuration after $o_2.\mathtt{remove()}$ has completed and the control has returned to the client, $I$.
At $u'$, $i_2$ still has the same version ($\mathit{i_2.iver}$) as $s$, ($\mathit{s.sver}$), but $i_1$ has a different version now. 
Thus, $i_1$ cannot traverse or remove an element of $s$ any more.

\begin{figure}[t]%
\centering
\input{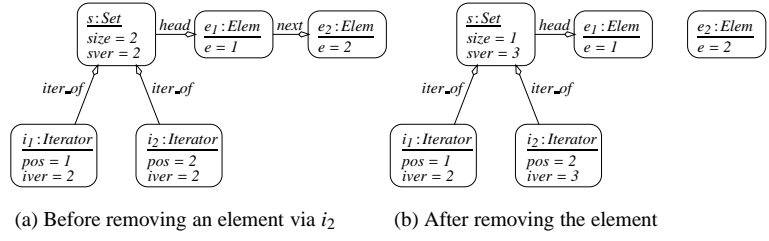}
\caption{Two configurations of set and iterator package}
\label{fig:iter_conc}
\end{figure}

\end{examp}

\section{Dynamic Package Interface (DPI)}\label{sect:formalinterface}

For a package $P$, its \emph{dynamic package interface} is essentially a set of \emph{nested object graphs} representing heap configurations together with a set of \emph{object mappings} over them, one for each distinct method invocation.

Each nested object graph represents an unbounded number of heap configurations. 
An object mapping for a method invocation specifies how the objects of a source heap configuration are transformed to the objects of a destination heap configuration.
Object mappings use an extended notion of object graphs with \emph{role labelling} to identify the callee and the arguments of the method calls.
Up to isomorphism, the set of object mappings of a DPI specify the effect of all possible public method calls on distinct heap configurations of a package.

In the remainder of this section, in Section~\ref{sect:absgraph}, we present the notions of nested object graphs and cast nested object graphs, followed by the notion of object mapping, in Section~\ref{sect:objmapping}. 
In Section~\ref{sect:dpiformal}, we present DPI formally.

\subsection{Nested Object Graphs}\label{sect:absgraph}

A \emph{nested object graph} $H$ over a package $P$ is a tuple $(\objectlabel,\funcpreds,\ao,\labelao,\ar,\nestingao)$ with
\begin{itemize}
\item $\objectlabel$ and $\funcpreds$: sets of \emph{object labels}  and \emph{reference fields}, respectively, 
\item $\ao$: a set of \emph{object nodes} identifiers,
\item $\ar: (\ao \times \funcpreds) \pfun \ao$ the \emph{reference edge} function,
\item $\labelao\!:\! \ao \rightarrow \objectlabel$ the \emph{object labelling} function,
\item $\nestingao\!:\! \ao \rightarrow \mathbb{N}_0$, the \emph{nesting level} function.
\end{itemize}
 
\noindent We call an object node with nesting level zero an \emph{object instance} and otherwise call it an \emph{abstract object}.
An abstract object represents an unbounded number of object instances. 
If an object node is connected via a reference label to another object node in $\ar$, it means that one or more object instances (depending on their relative nesting levels) in the source node have reference attributes pointing to an object instance in the destination node. 
We denote by $\classFromLabel$ the function from $\objectlabel$ to $P$ that extracts the class information from a label.

\noindent A nested object graph is \emph{well-formed} if:
$\forall (o_1,r,o_2), \ar(o_1, r) = o_2 \Rightarrow \nestingao(o_1) \geq \nestingao(o_2)$.
This constraint is necessary because it should not be possible for an object instance to reference more than one object instance with the same reference attribute.

\begin{examp}
Let us consider the graph in Figure \ref{fig:vlcover1}, which is a nested object graph.
Let the object node labelled with $\mathit{\preobj{V_{\nd}}Viewer}$ be denoted by $x$, then 
$V_{\nd}$ is the identifier that we use to refer to $x$ in the description, and we have $\mathit{\labelao(x) = Viewer}$, which tells the class of $x$ and the predicates and their valuation (none in this case). Finally, we have $\mathit{\nestingao(x) = 2}$. 
\end{examp}


A \emph{cast nested object graph} $G$ over $P$ is a tuple $(\objectlabel,\roles,\funcpreds,\ao,\labelao,\ar,\nameao,\nestingao)$ where
\begin{itemize}
\item $(\objectlabel,\funcpreds,\ao,\labelao,\ar,\nestingao)$ is a nested object graph over $P$,
\item $\roles$ is a set of \emph{object role} labels, and
\item $\nameao\!:\! \roles \rightarrow \ao$ is a \emph{role name} function.
\end{itemize}
An object of a cast nested object graph may have a role name in addition to its label.
A role name indicates the fixed responsibility of the object instance during a method call. 

A cast nested object graph can be obtained from a nested object graph by unfolding the graph and adding a role function.
The unfolding step copies a subgraph with nesting level greater than $0$ and decreases the nesting level of the copy by one.
This process is repeated until all the roles can be assigned to object instances.

A cast nested object graph is \emph{well-formed} if its role name function is injective: $\forall r_1,r_2 \in \roles,~ \nameao(r_1) = \nameao(r_2) \Rightarrow r_1 = r_2$.
%
Henceforth, we consider only well-formed nested object graphs and well-formed cast nested object graphs. 
We denote the set of all nested object graphs and the set of all cast nested object graphs over $P$ as $\mathcal{H}_P$ and $\mathcal{G}_P$, respectively. 

In our analysis, each cast nested object graph $G \in \mathcal{G}_P$ corresponds to a unique nested object graph $H \in \mathcal{H}_H$, as we will see in the next section. 
We assume the \emph{source} function $src\!:\! \mathcal{G}_P \rightarrow \mathcal{H}_P$, which determines the nested object graph of a cast nested object graph.
 
\begin{examp}
Let us consider the graph inside the box in the left hand side of Figure \ref{fig:vlt13}, which is a cast nested object graph whose source is $H_0$ in Figure \ref{fig:vlcover1}.
Let the object node labelled with $\mathit{\ncallee\!:\!\preobj{V_{\nd}}Viewer}$ be $x$, then 
$\mathit{\labelao(x) = Viewer}$,  $\mathit{\nestingao(x) = 0}$, and $\mathit{\nameao(x) = \ncallee}$. 
\end{examp}

The DPI shows the state of the system (i.e., the package together with its most general client) at the call and return points of public methods in the package.
In those states, the stack of the client is empty and $\This$ always refers to the most general client. Therefore, we omit this information in nested object graphs.
The roles in abstract graphs can be seen as a projection of the internal state of the most general client on the objects in the heap.
That is, the object instance of the most general client itself is not represented as a node in the graphs.

\subsection{Object Mapping}\label{sect:objmapping}

\paragraph{Notation.} For a package $P$, we denote by $\method_P$  the set of all its public methods: $\method_P = \bigcup_{C \in P} C.M_p$. 
For a public method $m(C_1,\cdots,C_n)$ of a class $C$, we define its
\emph{signature} as $\sig(m) = \{(C,\mathit{\ncallee}), (C_1,\narg{0}),\cdots,(C_n,\narg{n})\}$.

An \emph{object mapping} of a method $m \in \method_P$ is a tuple $(m,G,G',k)$ where $G,G' \in \mathcal{G}_P$, $k \subseteq G.\ao \times G'.\ao$ is a relation, and the following conditions are satisfied: 
\begin{itemize}
\item $G$ includes object instances for $\sig(m)$: 
$$
\forall (C,s) \in \sig(m),~ \exists o \in G.\ao,~ \classFromLabel(G.\labelao(o)) = C \wedge G.\nameao(o) = s;
$$
\item $\mathit{dom(k)} = G.\ao$; 

\item $k$ preserves the class of an object: $\forall (o_1,o_2)\in k,~ \classFromLabel(G.\labelao(o_1)) = \classFromLabel(G'.\labelao(o_2))$;

\item $k$ is functional on object instances: $\forall (o_1,o_2), (o_1,o_3) \in k,~ G.\nestingao(o_1) = 0 \Rightarrow o_2 = o_3$;

\item $k$ preserves the nesting level of object instances: \\
$\forall (o_1,o_2)\in k,~ G.\nestingao(o_1) = 0 \Leftrightarrow G'.\nestingao(o_2) = 0$;

\item $k$ preserves the role names of object instances: \\
$\forall (o_1,o_2)\in k,~ G.\nestingao(o_1) = 0 \Rightarrow G.\nameao(o_1)= G'.\nameao(o_2)$. 

\end{itemize}

For a set $M \subseteq \mathcal{G}_P$, by $\mathit{Maps}_P(M)$ we denote the set of all object mappings $(m,G,G',k)$ of package $P$ such that $G,G' \in M$.

An object mapping is a compact representation of the effect that a method call has on the objects of a package.
The mapping specifies how objects are transformed by the method call.
A pair $(o_1,o_2) \in k$ indicates that each concrete object represented by the abstract object $o_1$ might become part of the target abstract object $o_2$.
The total number of concrete objects is always preserved.
Because nested object graphs can represent more than one concrete state, there can be more than one object mapping associated with a given method call and source graph, as well as multiple target objects for each source object in the source graph of one object mapping.

\begin{examp}
Let us consider the two cast nested object graphs inside the boxes in the left and right hand side of Fig.~\ref{fig:vlt13}. 
Denote these two graphs by $G$ and $G'$.
Figure~\ref{fig:vlt13} then represents the object mapping: $(\mathtt{set},G,G',\{(V,V),(L_a,L_a),(L,L),(V_*,V_*)\}).$
\end{examp}

Note that in addition to $\ncallee$ and $\narg{0}$ role names, the object mapping in Figure~\ref{fig:vlt13} also uses $\nscope{0} \in G.\roles$, which labels an object instance that is not part of the signature of the method. 
The $\nscope{i}$ role names are used to label all such object instances.
One last type of role names that are used by object mappings is $\nnew{i}$ role names, which label the objects that are created by a method call.
To improve the readability of some figures we omit abstract objects that are not modified.
We show only the objects part of the connected component affected by the call.

\subsection{Definition: DPI}\label{sect:dpiformal} 

A DPI of a package $P$ is a tuple $(\absheap,\conheap,\maps,\errheap)$ where
\begin{itemize}
\item $\absheap \subseteq \mathcal{H}_P$ is a finite set of nested object graphs,
\item $\conheap \subseteq \mathcal{G}_P$ is a finite set of cast nested object graphs,
\item $\maps \subseteq \mathit{Maps}_P(\conheap)$ the set of \emph{object mappings}; and
\item $\errheap \subseteq \absheap$ the set of \emph{error} nested object graphs.
\end{itemize}
The DPI $(\absheap,\conheap,\maps,\errheap)$ is \emph{well-formed} if: 
\begin{enumerate}
\item the castgraphs come from $\absheap$: $\forall G \in \conheap,~ \srca{G} \in \absheap$

\item it is \emph{safe}: $\forall (m,G,G') \in \maps,~ \srca{G} \in (\absheap - \errheap)$; and

\item it is \emph{complete} in that a non-error covering nested object graph has a mapping for all methods: 
$$
\begin{array}{c}
\forall H \in (\absheap - \errheap),~ \forall o \in H.\ao,~ \forall m \in \classFromLabel(G.\labelao(o)).M_p,~ \exists (m,G,G') \in \maps,~ \srca{G} = H.
\end{array}
$$
\end{enumerate}

Well-formed DPIs characterize the type of interface that we are interested in computing for OO packages. 
Following the analogy between a DPI and an FSM, the set 
of nested object graphs correspond to the ``states'' of the state machine and the set of object mappings correspond to the ``transitions''. 
Section~\ref{sect:compdpi} describes how a well-formed DPI can be computed for a package soundly via an abstract semantics that simulates the concrete semantics of Section~\ref{sect:lang}.
Henceforth by a DPI, we mean a well-formed DPI.

A DPI can be understood in two ways.
The first interpretation comes directly from the abstract OO program semantics of Section~\ref{sect:compdpi}.
The second interpretation views the DPI as a counter program.
In this program each $H \in \absheap$ has a control location and for each node in $H.O$ there is a counter variable.
The value of a counter keeps track of the number of concrete objects that are represented by the corresponding abstract object node.
Object mappings can be translated into updates of the counters.
Further details of that interpretation can be found in Section~\ref{sect:compsubsect} and \cite{StructuralCounterAbs}.

\section{Abstract Semantics for Computing DPI}\label{sect:compdpi}
In this section, we present the abstraction layers that we use to compute the DPI of a package.
Section~\ref{sec:abstractsem} presents our \emph{depth-bounded abstract} domain, which ensures that any chain of objects of a package has a bounded depth when represented in this domain.
Section~\ref{sec:analysis} presents our \emph{ideal abstract} domain, which additionally ensures that any number of objects of a package are represented finitely.
Section~\ref{sect:compsubsect} describes how the DPI of a package can be computed by encoding the ideal abstract interpretation of a package as a numerical program.

\subsection{Preliminaries}

For a transition system $\LTS = (X,X_0,\rightarrow)$,
we define the \emph{post operator} 
as $\post.\LTS: \powerset(X) \rightarrow \powerset(X)$ with 
$\post.\LTS(Y) = \pset{x' \in X}{\exists x \in Y.\, x   \rightarrow  x'}$. 
The \emph{reachability set} of $\LTS$, denoted
$\Reach(\LTS)$, is defined by $\Reach(\LTS) = \lfp^\subseteq(\lambda
Y. X_0 \cup \post.\LTS(Y))$.

A \emph{quasi-ordering} $\leq$ is a reflexive and transitive relation
$\leq$ on a set $X$. In the following $X(\leq)$ is a quasi-ordered
set. 
The \emph{downward closure} (resp.\ \emph{upward closure}) 
of $Y \subseteq X$ is $\dc{Y} = \pset{x \in X}{\exists y \in Y.\, x \leq y}$
(resp.\ $\uc{Y} = \pset{x \in X}{\exists y \in Y.\, y \leq x}$).
A set $Y$ is \emph{downward-closed} (resp.\ \emph{upward-closed}) if $Y = \dc{Y}$ (resp.\ $Y=\uc{Y}$). 
An element
$x \in X$ is an \emph{upper bound} for $Y \subseteq X$ if for all $y \in Y$ we have $y \leq x$. 
A nonempty set $D \subseteq X$ is \emph{directed} if any two
elements in $D$ have a common upper bound in $D$. 
A set $I \subseteq X$ is an \emph{ideal} of $X$ if $I$ is downward-closed and
directed. 
%
A quasi-ordering $\leq$ on a set $X$ is a \emph{well-quasi-ordering}
(wqo) if any infinite sequence $x_0,x_1,x_2,\ldots$ of elements from
$X$ contains an increasing pair $x_i \leq x_j$ with $i < j$. 

%
%
A \emph{well-structured transition system} (WSTS) is a tuple $ \LTS =
(X,X_0,\rightarrow,\leq)$ where
$(X,X_0,\rightarrow)$ is a transition system and $\leq \; \subseteq X \times X$ is a wqo that 
is \emph{monotonic} with respect to $\rightarrow$, i.e.,
for all $x_1,x_2,y_1,t$ such that $x_1 \leq
y_1$ and $x_1 \rightarrow x_2$, there exists $y_2$ such that $y_1
\rightarrow y_2$ and $x_2 \leq y_2$. The \emph{covering set} of a
well-structured transition system $\LTS$, denoted $\Cover(\LTS)$, is
defined by $\Cover(\LTS) = \dc \Reach(\LTS)$.

\subsection{Depth-Bounded Abstract Semantics}
\label{sec:abstractsem}

We now present an abstract semantics for OO programs. 
Given an OO program $S$, our abstract semantics of $S$ is a transition system 
$S_\heap^\# = (\absconfs,\absinit,\rightarrow^\#_\heap)$ that is 
obtained by an abstract interpretation~\cite{CousotCousot79ProgramAnalysisFrameworks} of $S$. 
Typically, the system $S_\heap^\#$ is still an infinite state system. 
However, the abstraction ensures that $S_\heap^\#$ belongs to the 
class of \emph{depth-bounded systems}~\cite{Meyer08OnBoundednessInDepth}. 
Depth-bounded systems are well-structured transition systems that can be effectively 
analyzed~\cite{DBLP:conf/fossacs/WiesZH10}, and this will enable us to compute
the dynamic package interface. 


%

\noindent
\emph{Heap Predicate Abstraction.}
We start with a heap predicate abstraction, following shape analysis \cite{Sagiv02:Shape,Podelski05:Boolean}.
Let $\unarpreds$ be a finite set of \emph{unary abstraction predicates} from $\Pred.(\set{\var} \cup \classes.A)$ where $\var$ is a fresh 
variable different from $\idthis$ and $\idnull$. 
For a configuration $u = (\objects,\cdot,\store)$ and $o \in \objects$, we write $u \models p(o)$ iff $\store[x\mapsto o](p)=\booltrue$.
Further, let $\funcpreds$ be a subset of the reference fields in $\classes.A$. 
We refer to $\funcpreds$ as \emph{binary abstraction predicates}. 
For an object $o \in \allobjects$, we denote by $\funcpreds(o)$ the set 
$\funcpreds \cap \class(o).A$. 

The concrete domain $D$ of our abstract interpretation is the powerset of configurations 
$D = \powerset(\conconfs)$, ordered by subset inclusion. 
The abstract domain $\absdom_\heap$ is the powerset of \emph{abstract configurations} $\absdom_\heap = \powerset(\absconfs)$, 
again ordered by subset inclusion. 
An abstract configuration $u^\# \in \absconfs$ is like a concrete configuration except that the store is abstracted by a finite labelled graph, 
where nodes are object identifiers, edges correspond to the values of reference fields in $\funcpreds$, and 
node labels denote the evaluation of objects on the predicates in $\unarpreds$. 
That is, the abstract domain is parameterized by both $\unarpreds$ and $\funcpreds$.

Formally, an abstract configuration $u^\# \in \absconfs$ is a tuple $(\objects,\This,\curq,\nu,\pv,\absstore)$ where 
$\objects \subseteq \allobjects$ is a finite set of object identifiers, 
$\This \in \objects$ is the current object, $\curq \in F.Q$ is the current control location,
$\nu$ is a finite sequence of triples $(o,x,q)$ of objects, variables, and control location,
$\pv: \objects \times \unarpreds \to \BB$ is a \emph{predicate valuation}, and 
$\absstore$ is an \emph{abstract store} that maps objects in $o \in \objects$ and reference 
fields $a \in \funcpreds(o)$ to objects $\absstore(p,a) \in \objects$.
Note that we identify the elements of $\absconfs$ up to isomorphic renaming of object identifiers.

The meaning of an abstract configuration is given by a concretization function $\gamma_\bool: \absconfs \to D$ defined as follows: 
for $u^\# \in \absconfs$ we have $u \in \gamma_\bool(u^\#)$ iff 
(i)
$u^\#.\objects = u.\objects$; 
(ii) $u^\#.\This = u.\This$; 
(iii) $u^\#.\curq = u.\curq$; 
(iv) $u^\#.\nu = u.\nu$; 
(v) for all $o \in u.\objects$ and $p \in \unarpreds$, $u^\#.\pv(o,p) = \booltrue$ iff $u \models p(o)$; and 
(vi) for all objects $o \in \objects$, and $a \in \funcpreds(o)$, $u.\store(o,a) = u^\#.\absstore(o,a)$. 
We lift $\gamma_\bool$ pointwise to a function $\gamma_\bool : \absdom_\heap \to D$ by defining $\gamma_\bool(U^\#) = \bigcup \pset{\gamma_\bool(u^\#)}{u^\# \in U^\#}$. 
Clearly, $\gamma_\bool$ is monotone. 
It is also easy to see that $\gamma_\bool$ distributes over meets because for each configuration $u$ 
there is, up to isomorphism, a unique abstract configuration $u^\#$ such that $u \in \gamma_\bool(u^\#)$. 
Hence, let $\alpha_\bool: D \to \absdom_\heap$ be the unique function such that $(\alpha_\bool,\gamma_\bool)$ forms a Galois connection 
between $D$ and $\absdom_\heap$, i.e., $\alpha_\bool(U) = \bigcap \pset{U^\#}{U \subseteq \gamma_\bool(U^\#)}$.

The abstract transition system $S_\heap^\#= (\absconfs,\absinit,\rightarrow^\#_\heap)$ is
obtained by 
setting $\init^\# = \alpha_\bool(\init)$ and 
defining $\rightarrow^\#_\heap\; \subseteq \absconfs \times \absconfs$ as follows. 
Let $u^\#,v^\# \in \absconfs$.
We have $u^\# \rightarrow^\#_\heap v^\#$ iff $v^\# \in \alpha_\bool \circ \post.S \circ \gamma_\bool (u^\#)$. 

\begin{theorem}
The system $S_\heap^\#$ simulates the concrete system $S$, i.e., 
(i) $\init \subseteq \gamma_\heap(\absinit)$ and 
(ii) 
for all $u,v\in\conconfs$ and $u^\# \in\absconfs$,
if $u\in\gamma_\heap(u^\#)$ and $u\rightarrow v$,
then there exists $v^\# \in\absconfs$ such that $u^\#\rightarrow^\#_\heap v^\#$ and $v\in\gamma_\heap(v^\#)$.
\end{theorem}

\begin{proof}
\emph{(Sketch)}
We can use the framework of abstract interpretation~\cite{CousotCousot77AbstractInterpretation} to prove the theorem.
By definition, $(\alpha_\bool,\gamma_\bool)$ forms a Galois connection between $D$ and $\absdom_\heap$.
Furthermore, $u^\# \rightarrow^\#_\heap v^\#$ iff $v^\# \in \alpha_\bool \circ \post.S \circ \gamma_\bool (u^\#)$. 
\end{proof}

\noindent\emph{Depth-Boundedness.}
Let $u^\# \in \absconfs$ be an abstract configuration.
A \emph{simple path} of length $n$ in $u^\#$ is a sequence of distinct objects $\pi=o_1,\dots,o_n$ 
in $u^\#.\objects$ such that for all $1 \leq i < n$, 
there exists $a_i$ with $u^\#.\absstore(o_i,a_i)=o_{i+1}$ or $u^\#.\absstore(o_{i+1},a_i)=o_i$ (the path is not directed). 
We denote by $\lsp(u^\#)$ the length of the longest simple path of $u^\#$.
We say that a set of abstract configurations $U^\# \subseteq \absconfs$  is \emph{depth-bounded} 
if $U^\#$ is bounded in the length of its simple paths, i.e., there exists $k \in \NN$ such that $\forall u^\# \in U^\#, \lsp(u^\#) \leq k$ and the size of the stack $|u^\#.\nu| \leq k$.

We show that under certain restrictions on the binary abstraction predicates $\funcpreds$, 
the abstract transition system $S_\heap^\#$ is a well-structured transition system. 
For this purpose, we define the \emph{embedding order} on abstract configurations.
An \emph{embedding} for two configurations $u^\#, v^\#: \absconfs$ is a 
function $h: u^\#.\objects \to v^\#.\objects$ such that the following conditions hold: 
(i) $h$ preserves the class of objects: for all $o \in u^\#.\objects$, $\class(o)=\class(h(o))$;
(ii) $h$ preserves the current object, $h(u^\#.\This)=v^\#.\This$;
(iii) $h$ preserves the stack, $\bar{h}(u^\#.\nu)=v^\#.\nu$ where  $\bar{h}$ is the unique extension of $h$ to stacks;
(iv) $h$ preserves the predicate valuation: 
for all $o \in u^\#.\objects$ and $p \in \unarpreds$, $u^\#.\pv(o,p)$ iff $v^\#.\pv(h(o),p)$; and 
(v) $h$ preserves the abstract store, i.e., for all $o \in u^\#.\objects$ and $a \in \funcpreds(o)$, 
we have $h(u^\#.\store^\#(o,a)) = v^\#.\store^\#(h(o),a)$. 
The embedding order $\preceq: \absconfs \times \absconfs$ is then as follows: 
for all $u^\#, v^\#: \absconfs$, $u^\# \preceq v^\#$ iff $u^\#$ and $v^\#$ share the same current control 
location ($u^\#.\curq = v^\#.\curq$) and there exists an injective embedding of $u^\#$ into $v^\#$.

\begin{lemma}
(1) The embedding order is monotonic with respect to abstract transitions in
$S^\#_\heap=(\absconfs,\absinit,\rightarrow^\#_\heap)$\techreport{, 
i.e., for all $u_1^\#,u_2^\#,v_1^\# \in \absconfs$, if 
$u_1^\# \preceq u_2^\#$ and $u_1^\# \rightarrow^\#_\heap v_1^\#$, then there exists $v_2^\# \in \absconfs$ such that $v_1^\# \preceq v_2^\#$ and $u_2^\# \rightarrow^\#_\heap v_2^\#$}{}.
(2)  Let $U^\#$ be a depth-bounded set of abstract configurations. Then $(U^\#,\preceq)$ is a wqo.
\label{lem-dpi-monotone}
\end{lemma}

\begin{proof}
The first part follows form the definitions.
For the second part, we can reduce it to the result from~\cite{StructuralCounterAbs}.
We just need to encode the stack into the graph.
The stack itself can be easily encoded as a chain with special bottom and top node.
The assumption that the stack is bounded guarantees that can still apply \cite[Lemma~2]{StructuralCounterAbs}.
\end{proof}

If the set of reachable configurations of the abstract transition system $S_\heap^\#$ is depth-bounded, 
then $S_\heap^\#$ induces a well-structured transition system.

\begin{theorem}
If $\Reach(S_\heap^\#)$ is depth-bounded, then $(\Reach(S^\#), \absinit, \rightarrow^\#_\heap, \preceq)$ is a WSTS.
\end{theorem}

\begin{proof}
The theorem follows from Lemma~\ref{lem-dpi-monotone} and \cite[Theorem 2]{Meyer08OnBoundednessInDepth}.
\end{proof}

In practice, we can ensure depth-boundedness of $\Reach(S_\heap^\#)$ syntactically by choosing the set of binary 
abstraction predicates $\funcpreds$ such that it does not contain reference fields that span recursive 
data structures. 
Such reference fields are only allowed to be used in the defining formulas of the unary abstraction predicates.
Recursive data structures can be dealt with only if they are private to the package, i.e. not exposed to the user.
In that case the predicate abstraction can use a more complex domain that understand such shapes, e.g. \cite{Sagiv02:Shape}.
In the next section, we assume that the set $\Reach(S_\heap^\#)$ is depth-bounded and we identify $S_\heap^\#$ with its induced WSTS. 

\begin{examp}

Figure~\ref{fig:iter_dabs} depicts the two corresponding, depth-bounded abstract configurations of the concrete configurations in Figure~\ref{fig:iter_conc}.
The objects are labelled with their corresponding unary predicates. 
A labelled arrow between two objects specifies that the corresponding binary predicate between two object holds.
The set of unary abstraction predicates consists of:
\begin{equation*}
\begin{array}{lclclcl}
\mathit{empty(x)} & \equiv & x.\mathtt{size} = 0 & \quad &
\mathit{synch(x)} & \equiv & x.\mathtt{iver}=x.\mathtt{iter\_of.sver} \\
\mathit{mover(x)} & \equiv & x.\mathtt{pos} <x.\mathtt{iter\_of.size} &  \quad &
\mathit{positive(x)} & \equiv & x.\mathtt{e}>0
\end{array}
\end{equation*}
\noindent The set of binary abstraction predicates is $\funcpreds = \set{\mathtt{iter\_of}}$.
If we had also included $\mathtt{head}$ and $\mathtt{next}$ in $\funcpreds$, the resulting abstraction would not have been depth bounded.

\begin{figure}[t]%
\centering
\input{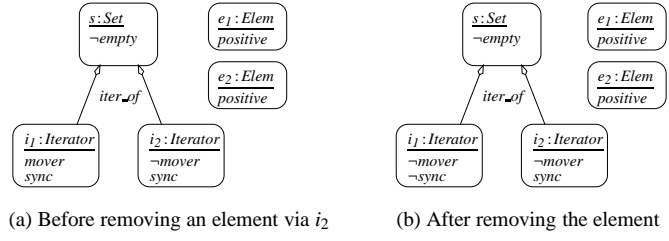}
\caption{Two depth-bounded abstract configurations}
\label{fig:iter_dabs}
\end{figure}

\end{examp}

\subsection{Ideal Abstraction}
\label{sec:analysis}
In our model, the errors are local to objects.
Thus, we are looking at the control-state reachability question.
This means that the set of abstract error configurations is upward-closed with respect to the embedding order $\preceq$, 
i.e., we have $\abserror = \uc\abserror$. 
From the monotonicity of $\preceq$ we therefore conclude that $\Reach(S_\heap^\#) \cap \abserror = \emptyset$ iff $\Cover(S_\heap^\#) \cap \abserror = \emptyset$. 
This means that if we analyze the abstract transition system $S_\heap^\#$ modulo downward closure of abstract configurations, this 
does not incur an additional loss of precision. 
We exploit this observation as well as the fact that $S_\heap^\#$ is well-structured to 
construct a finite abstract transition system whose configurations are given by downward-closed 
sets of abstract configurations. 
We then show that this abstract transition system can be effectively computed. 

Every downward-closed subset of a wqo is a \emph{finite} union of ideals. 
In previous work~\cite{Zufferey12:Ideal}, we formalized an abstract interpretation coined \emph{ideal abstraction}, which exploits this observation to obtain a generic terminating analysis for computing an over-approximation of the covering set of a WSTS. We next show that ideal abstraction applies to the depth-bounded abstract semantics by providing an appropriate finite representation of ideals and how to use it to compute the DPI.
The abstract domain $\absdom_\ideal$ of the ideal abstraction is given by downward-closed sets of abstract configurations, 
which we represent as finite sets of ideals. 
The concrete domain is $\absdom_\heap$. 
The ordering on the abstract domain is subset inclusion. 
The abstraction function is downward closure.

Formally, we denote by $\Idl(\absconfs)$ the set of all depth-bounded ideals of abstract configurations with respect to the embedding order.
Define the quasi-ordering $\sqsubseteq$ on $\powerset_\fin(\Idl(\absconfs))$ as the point-wise extension of $\subseteq$ from the 
ideal completion $\Idl(\absconfs)$ of $\absconfs(\preceq)$ to $\powerset_\fin(\Idl(\absconfs))$:
\[
\mathcal{I}_1 \sqsubseteq \mathcal{I}_2 \iff \forall I_1 \in \mathcal{I}_1.\, \exists I_2 \in \mathcal{I}_2.\, I_1 \subseteq I_2
\]
The abstract domain $\absdom_\ideal$ is the quotient of $\powerset_\fin(\Idl(\absconfs))$
with respect to the equivalence relation $\sqsubseteq \cap
\sqsubseteq^{-1}$. For notational convenience we use the same symbol
$\sqsubseteq$ for the quasi-ordering on $\powerset_\fin(\Idl(\absconfs))$ and
the partial ordering that it induces on
$\absdom_\ideal$. We further identify the elements of
$\absdom_\ideal$ with the finite sets of maximal ideals, i.e., for
all $L \in \absdom_\ideal$ and $I_1,I_2 \in L$, if $I_1 \subseteq
I_2$ then $I_1 = I_2$.
The abstract domain $\absdom_\ideal$ is defined as $\powerset_\fin(\Idl(\absconfs))$.
The concretization function $\gamma_\ideal : \absdom_\ideal \to \absdom_\heap$ is $\gamma_\ideal(\mathcal{I}) = \bigcup \mathcal{I}$.  
Further, define the abstraction function $\alpha_\ideal: \absdom_\heap \to \absdom_\ideal$ as 
$\alpha_\ideal(U^\#) = \pset{I \in \Idl(\absconfs)}{I \subseteq \dc U^\#}$. 
From the ideal abstraction framework~\cite{Zufferey12:Ideal}, it follows that $(\alpha_\ideal,\gamma_\ideal)$ forms a 
Galois connection between $\absdom_\heap$ and $\absdom_\ideal$. 
The overall abstraction is then given by the Galois connection $(\alpha,\gamma)$ between $D$ and $\absdom_\ideal$, 
which is defined by $\alpha = \alpha_\ideal \circ \alpha_\heap$ and $\gamma = \gamma_\heap \circ \gamma_\ideal$. 
We define the \emph{abstract post operator} $\abspost$ of $S$ as the most precise abstraction of $\post.S$ with respect to this 
Galois connection, i.e., $\abspost.S = \alpha \circ \post.S \circ \gamma$.

In the following, we assume the existence of a \emph{sequence widening operator} $\idealwiden: \Idl(\absconfs)^+ \pto \Idl(\absconfs)$, 
i.e., $\idealwiden$ satisfies the following two conditions: 
(i) \emph{covering condition}: for all $\mathcal{I} \in \Idl(\absconfs)^+$, if $\idealwiden(\mathcal{I})$ is defined, 
then for all $I$ in $\mathcal{I}$, $I \subseteq \idealwiden(\mathcal{I})$.; and 
(ii) \emph{termination condition}: for every ascending chain $(I_i)_{i \in \NN}$ in $\Idl(\absconfs)$, 
the sequence $J_0 = I_0$, $J_{i} = \idealwiden(I_0\dots I_i)$, for all $i > 0$, is well-defined and an ascending stabilizing chain. 
\techreport{We will define the operator $\idealwiden$ later in this section.}{}

The ideal abstraction induces a finite labeled transition system $S_\ideal^\#$ whose configurations are ideals of abstract configurations. 
There are special transitions labeled with $\epsilon$, which we refer to as \emph{covering transitions}. 
We call $S_\ideal^\#$ the \emph{abstract covering system} of $S_\heap^\#$. 
This is because the set of reachable configurations of $S_\ideal^\#$ over-approximates the covering set of $S_\heap^\#$, 
i.e., $\Cover(S_\heap^\#) \subseteq \gamma_\ideal(\Reach(S_\ideal^\#))$. 
Furthermore, the directed graph spanned by the non-covering transitions of $S_\ideal^\#$ is acyclic.

Formally, we define $S_\ideal^\#=(\idealconfs,\idealinit,\idealtransrel{\cdot})$ as follows.
The initial configurations $\idealinit$ are given by $\idealinit = \alpha_\ideal(\absinit)$. 
The set of configurations $\idealconfs \subseteq \Idl(\absconfs)$ and 
the transition relation $\idealtransrel{\cdot} \subseteq \idealconfs \times \idealconfs$ 
are defined as the smallest sets satisfying the following conditions: 
(1) $\idealinit \subseteq \idealconfs$; and 
(2) for every $I \in \idealconfs$, let $\mathit{paths}(I)$ be the set of all sequences of ideals $I_0 \dots I_n$ with $n \geq 0$ such that $I_0 \in \idealinit$, 
$I_n = I$, and for all $0 \leq i < n$, $I_i \idealtransrel{\cdot} I_{i+1}$. 
Then, for every path $\mathcal{I}=I_0\dots I_n \in \mathit{paths}(I)$, if 
there exists $i < n$ such that $I \subseteq I_i$, then $I \idealtransrel{\epsilon} I_i$. 
Otherwise, for all 
$I' \in \abspost.S \circ \gamma_\ideal (I)$, 
let $J' = \idealwiden(\mathcal{I}' I')$ where $\mathcal{I'}$ is the subsequence of all 
ideals $I_i$ in $\mathcal{I}$ with $I_i \subseteq I'$, then $J' \in \idealconfs$ and $I \idealtransrel{\cdot} J'$.

\begin{theorem}
  The abstract covering system $S_\ideal^\#$ is computable and finite.
\end{theorem}

\begin{proof}
\emph{(Sketch)}~
Following the result from \cite{Zufferey12:Ideal}, we can effectively compute an inductive overapproximation $\mathcal{C}$ of the covering set of $S_\ideal^\#$.
From \cite[Lemma~15]{DBLP:conf/fossacs/WiesZH10}, we have a finite representation of $\mathcal{C}$.
Finally, $\idealtransrel{\cdot}$ can be effectively computed as we will see in the remainder of the section.
\end{proof}

Define the relation $\idealtransrel{*}\; \subseteq \idealconfs \times \idealconfs$ as 
$\idealtransrel{*} \; = \; \idealtransrel{\cdot} \cup \idealtransrel{\epsilon} \circ \idealtransrel{\cdot}$. 
We now state our main soundness theorem. 
\begin{theorem} {\bf [Soundness]}
  The abstract covering system $S_\ideal^\#$ simulates $S$, i.e., (i) $\init \subseteq \gamma(\idealinit)$ and 
(ii) for all $I \in \idealconfs$ and $u,v \in \Reach(S)$, 
if $u \in \gamma(I)$ and $u \rightarrow v$, then there exists $J \in \idealconfs$ such that $v \in \gamma(J)$ and $I \idealtransrel{*} J$.
\end{theorem}

\begin{proof}
\emph{(Sketch)}~~
The abstract covering system is just a lifting of the original transition system to a finite-state system by partitioning the states into a finite number of sets given by the incomparable ideals in covering set or an overapproximation of it.
The lifting relies on the monotonicity property of the underlying WSTS to ensures simulation.
The transition relation $\idealtransrel{\cdot}$ maps states from ideal to ideal while ensuring that the target ideal contains at least one larger state.
\end{proof}

In the rest of this section we explain how we represent ideals of abstract configurations and how the operations for computing the abstract covering system are implemented.

\paragraph{Representing Ideals of Abstract Configurations.}

The ideals of depth-bounded abstract configurations are recognizable by regular hedge automata~\cite{DBLP:conf/fossacs/WiesZH10}. We can encode these automata into abstract configurations $\idealconf$ that are equipped with a \emph{nesting level function}. The nesting level function indicates how the substructures of the abstract store of $\idealconf$ can be replicated to obtain all abstract configurations in the represented ideal. 

Formally, a \emph{quasi-ideal configuration} $\idealconf$ is a tuple $(\objects, \This, \curq, \nu, \pv, \absstore, \nlf)$ where $\nlf: \objects \to \NN$ is the nesting level function and $(\objects, \This, \curq, \nu, \pv, \absstore)$ is an abstract configuration, except that $\pv$ is only a partial function $\pv: \objects \times \unarpreds \pto \BB$. We denote by $\allqidealconfs$ the set of all quasi-ideal configurations.
We call $\idealconf=(\objects, \This, \curq, \nu, \pv, \absstore, \nlf)$ simply \emph{ideal configuration}, if $\pv$ is total and for all $o \in \objects$, $a \in \funcpreds(o)$, $\nlf(o) \geq \nlf(\absstore(o,a))$.  We denote by $\absify{\idealconf}$ the \emph{inherent} abstract configuration $(\objects, \This, \curq, \nu, \pv, \absstore)$ of an ideal configuration $\idealconf$. Further, we denote by $\allidealconfs$ the set of all ideal configurations and by $\allfinidealconfs$ the set of all ideal configurations in which all objects have nesting level 0. We call the latter \emph{finitary} ideal configurations.

\paragraph{Meaning of Quasi-Ideal Configurations.}

An \emph{inclusion mapping} between quasi-ideal configurations $\idealconf=(\objects,\This,\curq,\nu,\absstore,\nlf)$ and $\oidealconf=(\objects',\This',\curq',\nu',{\absstore}', \nlf')$ is an embedding $h: \objects \to \objects'$ that satisfies the following additional conditions: (i) for all $o \in \objects$, $\nlf(o) \leq \nlf'(h(o))$; (ii) $h$ is injective with respect to level $0$ vertices in $\objects'$: for all $o_1,o_2 \in \objects$, $o' \in \objects'$, $h(o_1) = h(o_2) = o'$ and $\nlf'(o') = 0$ implies $o_1=o_2$; and (iii) for all distinct $o_1,o_2,o \in \objects$, if $h(o_1) = h(o_2)$, and $o_1$ and $o_2$ are both neighbors of $o$, then $\nlf'(h(o_1)) = \nlf'(h(o_2)) > \nlf'(h(o))$.

We write $\idealconf \leq_h \oidealconf$ if $\curq=\curq'$, and $h$ is an inclusion mapping between $\idealconf$ and $\oidealconf$. We say that $\idealconf$ is \emph{included} in $\oidealconf$, written $\idealconf \leq \oidealconf$, if $\idealconf \leq_h \oidealconf$ for some $h$.

We define the meaning $\den{\idealconf}$ of a quasi-ideal configuration $\idealconf$ as the set of all inherent abstract configurations of the finitary ideal configurations included in $\idealconf$:
\[
\den{\idealconf} = \pset{\absify{\oidealconf}}{\oidealconf \in \allfinidealconfs \land \oidealconf \leq \idealconf}
\]
We extend this function to sets of quasi-ideal configurations, as expected.

\begin{proposition}
  Ideal configurations exactly represent the depth-bounded   ideals of abstract configurations, i.e.,   $\pset{\den{\idealconf}}{\idealconf \in \allidealconfs} =   \Idl(\absconfs)$.
\end{proposition}

Since the relation $\leq$ is transitive, we also get:

\begin{proposition}
  For all $\idealconf,\oidealconf \in \allqidealconfs$, $\idealconf \leq \oidealconf$ iff $\den{\idealconf} \subseteq \den{\oidealconf}$.
\end{proposition}

It follows that inclusion of (quasi-)ideal configurations can be decided by checking for the existence of inclusion mappings, which is an NP-complete problem.

\techreport{
\paragraph{Folding and Unfolding of Quasi-Ideal Configurations.}

As we shall see, quasi-ideal configurations will be useful as an intermediate representation of the images of the abstract post operator. They can be thought of as a more compact representation of finite sets of ideal configurations. In order to be able to reduce quasi-ideal configurations to ideal configurations, we introduce several folding and unfolding operations.

First, we define a \emph{folded form} of quasi-ideal configurations. Let $\idealconf=(\objects,\This,\curq,\nu,\pv,\absstore,\nlf)$ be a quasi-ideal configuration. Intuitively, folding removes all objects from $\idealconf$ that have no effect on the denotation of $\idealconf$. Formally, we say that $\idealconf$ is in folded form if all inclusion mappings $h$ between $\idealconf$ and itself have the following property: for all distinct $o_1,o_2 \in \objects$, if $h(o_1)=h(o_2)$ then either $\nlf(o_1) < \nlf(o_2)$ and $\absstore(o,a)=o_1$ for some $o \in \objects$ with $\nlf(o) \leq \nlf(o_1)$, or conversly $\nlf(o_2) < \nlf(o_1)$ and $\absstore(o,a)=o_2$ for some $o \in \objects$ with $\nlf(o) \leq \nlf(o_2)$. Let $\fold$ be the function that transforms a quasi-ideal configuration into folded form.

Next, we introduce two operations on \emph{layers} of quasi-ideal configurations: \emph{copying} and \emph{unfolding} of a layer.
Let $\idealconf=(\objects,\This,\curq,\nu,\pv,\absstore,\nlf)$ and $\oidealconf=(\objects',\This',\curq',\nu',\pv',\absstore',\nlf')$ be quasi-ideal configurations.  For $i \geq 1$, denote all objects at nesting level $i$ or higher in $\idealconf$ by $\layer{\objects}{i} = \pset{ o \in \objects}{\nlf(o) \geq i}$. We call this set the \emph{$i$-th layer} of $\idealconf$.  Copying of the $i$-th layer involves duplicating the subgraph induced by $\layer{\objects}{i}$, unfolding in addition reduces the nesting level of all objects in the copy by $1$. 
Formally, assume $\oidealconf \leq_h \idealconf$ such that $h$ preserves the predicate valuation, i.e., for all $o' \in \objects'$, $\pv'(o')=\pv(h(o'))$. Further assume that there exists a partition $V,W_1,W_2$ of $\objects'$ such that $h$ is a bijection between $W_1$ and $\layer{\objects}{i}$, respectively, $W_2$ and $\layer{\objects}{i}$, respectively, $V$ and $\objects \setminus \layer{\objects}{i}$. If in addition, $h$ preserves the valuation of the nesting level function, i.e., for all $o \in \objects'$, $\nlf'(o')=\nlf(h(o'))$, then we say that $\oidealconf$ is obtained from $\idealconf$ by copying of the $i$-th layer under inclusion mapping $h$, written $\idealconf \copylayer{h,i} \oidealconf$. If on the other hand for all $o' \in V \cup W_1$, $\nlf'(o')=\nlf(h(o'))$ and for all 
$o' \in W_2$, $\nlf'(o')=\nlf(h(o'))-1$, then we say that $\oidealconf$ is obtained from $\idealconf$ by unfolding of the $i$-th layer under inclusion mapping $h$, written $\idealconf \unfoldlayer{h,i} \oidealconf$.

\paragraph{Reduction of Quasi-Ideal Configurations.}
We next describe how to reduce a quasi-ideal configuration to an equivalent set of ideal configurations.

Let $\idealconf=(\objects,\This,\curq,\nu,\pv,\absstore,\nlf)$ be a quasi-ideal configuration.  There are two reasons why $\idealconf$ may not be an ideal configuration: (1) for some object $o \in \objects$ and predicate $p \in \unarpreds$, $\pv(o,p)$ is undefined; and (2) for some object $o \in \objects$ and reference attribute $a \in \funcpreds(o)$, $\nlf(o) < \nlf(\absstore(o,a))$. Hence, we can recursively reduce $\idealconf$ into a finite set of ideal configuration $\mathcal{I}^\#$ as follows: define $\mathcal{I}^\#_0=\set{\idealconf}$ and for all $i \geq 0$, choose some $\oidealconf \in \mathcal{J}_0^\#$ that is not an ideal configuration, if such an element exists. Then obtain $\mathcal{I}_{n+1}^\#$ from $\mathcal{I}_n^\#$ by applying the following reduction steps:
\begin{enumerate}[(i)]
\item if there exists $o \in \objects$ and $a \in \funcpreds(o)$ with $\nlf(o) < \nlf(\absstore(o,a))$, choose one such $o$ and $a$. Let $i=\nlf(\absstore(o,a))$. Then replace $\oidealconf$ in $\mathcal{I}_{n}^\#$ with the set of all $\oidealconf_1 = (\objects_1,\This_1,\curq_1,\nu_1,\pv_1,\absstore_1,\nlf_1)$ such that $\oidealconf \unfoldlayer{h,i} \oidealconf_1$ and for the unique object $o' \in \objects_1$ with $h(o')=o$, $\nlf(\absstore'(o',a))=i-1$ holds.
\item else if there exists $o \in \objects$ with $\nlf(o) > 0$ and $p \in   \unarpreds$ such that $\pv(o,p)$ is undefined, choose one such $o$ and $p$. Then replace   $\oidealconf$ in $\mathcal{I}_{n}^\#$ with $\reduce(\oidealconf_2)$ where
$\oidealconf \copylayer{h,\nlf(o)} \oidealconf_1=(\objects_1,\This_1,\curq_1,\pv_1,\absstore_1,\nlf_1)$
with distinct $o_1,o_2 \in \objects_1$ such that $h(o_1)=h(o_2)=o$,
and $\oidealconf_2=(\objects_1,\This_1,\curq_1,\nu_1,\pv_1[(o_1,p) \mapsto \booltrue, (o_2,p) \mapsto \boolfalse], \absstore_1,\nlf_1)$;
\item else if there exists $o \in \objects$ with $\nlf(o)=0$ and $p \in \unarpreds$ such that $\pv(o,p)$ is undefined, choose one such $o$ and $p$. Then replace $\oidealconf$ in $\mathcal{I}_{n}^\#$ with $\oidealconf_1$ and $\oidealconf_2$, where $\oidealconf_1 = (\objects, \This, \curq,\nu, \pv[(o,p) \mapsto \booltrue],\absstore,\nlf)$ and $\oidealconf_2 = (\objects, \This, \curq,\nu, \pv[(o,p) \mapsto \boolfalse],\absstore,\nlf)$.
\end{enumerate}
If $\mathcal{I}_n^\#$ contains only ideal configurations, define $\mathcal{I}^\# = \mathcal{I}_n^\#$. The reduction sequence is finite. Let $\reduce:\allqidealconfs \to \powerset_\fin(\allidealconfs)$ be a function that performs the above reduction for some order of reduction steps. We extend $\reduce$ to finite sets of quasi-ideal configurations as follows: $\reduce(\mathcal{I}^\#) = \bigcup \pset{\reduce(\idealconf)}{\idealconf \in \mathcal{I}^\#}$.

\begin{proposition}
  Reduction does not change the meaning of quasi-ideal configurations, i.e., for all $\mathcal{I}^\# \subseteq_\fin \allqidealconfs$,   $\den{\mathcal{I}^\#} = \den{\reduce(\mathcal{I}^\#)}$.
\end{proposition}
}{
Quasi-ideal configurations are useful as an intermediate representation of the images of the abstract post operator. They can be thought of as a more compact representation of sets of ideal configurations. In fact, any quasi-ideal configuration can be reduced to an equivalent finite set of ideal configuration. We denote the function performing this reduction by $\reduce: \allqidealconfs \to \powerset_\fin(\allidealconfs)$ and we extend it to sets of quasi-ideal configurations, as expected. 
}

\begin{examp}

Figure~\ref{fig:iter_iabs} depicts the two corresponding, ideal abstract configurations of the two depth-bounded abstract configurations in Figure~\ref{fig:iter_dabs}.
The nesting level of each object is shown by the number 		next to it.
When the abstract configurations in Figure~\ref{fig:iter_dabs} are considered as finitary ideal configurations, then they are included in their corresponding ideal configurations in Figure~\ref{fig:iter_iabs}. 
The two inclusion mappings between the corresponding configurations in Figure~\ref{fig:iter_dabs} and Figure~\ref{fig:iter_iabs} are $\set{(i_1,i_1^\#),(i_2,i_2^\#),(s,s^\#),(e_1,e^\#),(e_2,e^\#)}$.

Note that since the nesting level of $\mathit{s^\#\!:\!Set}$ in both ideal configurations is zero, it is not possible to define inclusion mapping when there are more than one concrete set object. However, if the nesting levels of the set and iterator objects are incremented, then such an inclusion mapping can be defined.

\begin{figure}[t]%
\centering
\input{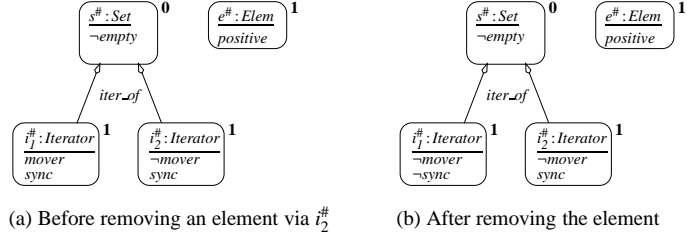}
\caption{Two ideal abstract configurations}
\label{fig:iter_iabs}
\end{figure}

\end{examp}

\paragraph{Computing the Abstract Post Operator.}

We next define an operator $\reppost.S$ that implements the abstract post operator $\abspost.S$ on ideal configurations. 
In the following, we fix an ideal configuration $\idealconf = (\objects,\This,\curq,\nu,\absstore,\nlf)$ and a transition $t=(\curq,\op,q')$ in $S$. 
For transitions not enabled at $\idealconf$, we set $\reppost.S.t(\idealconf)=\emptyset$.

We reduce the computation of abstract transitions $\absify{\idealconf} \rightarrow u^\#$ to reasoning about logical formulas. 
For efficiency reasons, we implicitly use an additional Cartesian abstraction~\cite{Ball03:Boolean} 
in the abstract post computation that reduces the number of required theorem prover calls.
For a set of variables $X$, we \techreport{define the}{assume a} 
\emph{symbolic weakest precondition} operator $\wlp: \Op.(\classes.A) \times \Pred.(X \cup \classes.A) \to \Pred.(X \cup \classes.A)$ \techreport{as follows. Let $\op \in \Op.(\classes.A)$ and $p \in \Pred.(X \cup \classes.A)$, then depending on the type of $\op$ we define: (i) $\wlp(\op,p)=p[e/\idthis]$, if $\op = \idthis := e$; (ii) $\wlp(\op,p)=p[(\lambda y.\, \ite{r = y}{y.x}{e})/x]$, if $\op = r.x := e$; (iii) $\wlp(\op,p)=p[(\lambda y.\, \ite{r = y}{y.x}{z})/x]$ where $z \in X$ is fresh, if $\op = \havoc{r.x}$; and (iv)  $\wlp(\op,p)=(q \Rightarrow p)$, if $\op = \assume{q}$}{that is defined as usual}.
In addition, we need a symbolic encoding of abstract configurations into logical formulas. For this purpose, define a function $\Gamma: \objects \to \Pred.(\objects \cup \classes.A)$ as follows: given $o \in \objects$, let $\objects(o)$ be the subset of objects in $\objects$ that are transitively reachable from $o$ in the abstract store $\absstore$, then $\Gamma(o)$ is the formula
\begin{align*}
& \Gamma(o) = 
\mathrm{distinct}(\objects(o) \cup \objects(\This)) \land \idthis=\This \land \idnull=\Null \land {}\\
&
\quad \bigwedge_{o' \in O(o) \cup \objects(\This)} \left(
\bigwedge_{p \in \unarpreds} \pv(o',p) \cdot p(o') \land \bigwedge_{a \in \funcpreds(o')} o'.a = \absstore(o'.a) \right) \\
& \text{where} \;
\pv(o',p) \cdot p(o') = 
\begin{cases}
  p(o') & \text{if $\pv(o',p)=\booltrue$}\\
  \neg p(o') & \text{if $\pv(o',p)=\boolfalse$}.
\end{cases}
\end{align*}
Now, let $\mathcal{J}^\#$ be the set of all quasi-ideal configurations $\oidealconf=(\objects,\This, q',\nu,\pv',{\absstore}',\nlf)$ that satisfy the following conditions:
\begin{itemize}
\item $\Gamma(\This) \wedge q$ is satisfiable, if $\op = \assume{q}$;
\item for all $o \in \objects$, $p \in \unarpreds$,
  if $\Gamma(o) \models \wlp(\op,p(o))$, then $\pv'(o,p)=\booltrue$, else if $\Gamma(o) \models \wlp(\op,\neg p(o))$, then $\pv'(o,p)=\boolfalse$, else $\pv'(o,p)$ is undefined;
\item for all $o,o' \in \objects$, $a \in \funcpreds(o)$,
  if $\Gamma(o) \land \Gamma(o') \models \wlp(\op,o.a=o')$, then ${\absstore}'(o,a)=o'$, else if $\Gamma(o) \land \Gamma(o') \models \wlp(\op, o.a \neq o')$, then ${\absstore}'(o,a) \neq o'$.
\end{itemize}
Then define $\reppost.S.t(\idealconf) = \reduce(\mathcal{J^\#})$.


\techreport{
Note that the size of $\mathcal{J}^\#$ is linear in the size of $\objects$ because at most one reference field in $\funcpreds$ changes its value due to $\op$ for at most one object in $\objects$. Moreover, $\mathcal{J}^\#$ can be computed from $I^\#$ with $\mathcal{O}(m n^2 +k n)$ theorem prover calls implementing the entailment checks, where $n$ is the size of $\objects$, $m$ the number of binary abstraction predicates $\funcpreds$, and $k$ the number of unary abstraction predicates $\unarpreds$.
}{}

\techreport{
\paragraph{Widening of Ideals.}

We use a widening operator that exploits the monotonicity of the heap abstraction $S_\heap^\#$ by mimicing acceleration in acceleration-based forward analyses of WSTS. 

Let $\idealconf=(\objects,\This,\curq,\nu,\pv,\absstore,\nlf)$ and $\oidealconf=(\objects',\This',\curq',\nu',\pv',\absstore',\nlf')$ be ideal configurations such that $\idealconf \leq \oidealconf$ and let $h$ be an inclusion mapping between $\idealconf$ and $\oidealconf$. Intuitively, we extrapolate from $\idealconf$ and $\oidealconf$ by increasing the nesting level of all objects in $\oidealconf$ that are not included in $\idealconf$. Formally, define the \emph{extrapolation} $\extrapol(\idealconf,\oidealconf,h)$ as the ideal configuration $(\objects',\This',\curq',\nu',\pv',\absstore',\nlf'')$ where the new nesting level function $\nlf''$ is defined as follows:
for all objects $o \in \objects'$, if there is no $o' \in h(\objects)$ such that $o'$ is transitively reachable from $o$ in $o'$ and $\nlf'(o') = \nlf'(o)$, then $\nlf''(o)=\nlf'(o)+1$, otherwise $\nlf''(o) = \nlf'(o)$. Note that we have $\idealconf \leq \oidealconf \leq \extrapol(\idealconf, \oidealconf, h)$.

We extend the extrapolation operator to the sequence widening operator $\idealwiden$ as follows: let $\mathcal{I}^\# = \idealconf_0 \dots \idealconf_n$ be a sequence of ideal configurations with $n \geq 0$ such that for all $i < n$, $\idealconf_i \leq \idealconf_n$ for some $h_i$. Then define $\oidealconf_n = \idealconf_n$ and for all $i < n$, $\oidealconf_i = \extrapol(\idealconf_i, \oidealconf_{i+1}, h_i)$ for some choice of the inclusion mapping $h_i$ between $\idealconf_i$ and $\oidealconf_{i+1}$. Then define $\idealwiden(\mathcal{I}^\#)=\oidealconf_0$.
In all other cases let $\idealwiden(\mathcal{I}^\#)$ be undefined.

\begin{proposition}
  $\idealwiden$ is a sequence widening operator for depth-bounded sequences of ideal configurations.
\end{proposition}
}{}

\subsection{Computing the Dynamic Package Interface}\label{sect:compsubsect}

We now describe how to compute the dynamic package interface for a given package $P$. The computation proceeds in three steps. First, we compute the OO program $S=(P,I)$ that is obtained by extending $P$ with its most general client $I$. Next, we compute the abstract covering system $S_\ideal^\#$ of $S$ as described in Sections~\ref{sec:abstractsem} and \ref{sec:analysis}. We assume that the user provides sets of unary and binary abstraction predicates $\unarpreds$, respectively, $\funcpreds$ that define the heap abstraction. Alternatively, we can use heuristics to guess these predicates from the program text of the package. For example, we can add all branch conditions in the program description as predicates.
Finally, we extract the package interface from the computed abstract covering system.
We describe this last step in more detail.

We can interpret the abstract covering system as a numerical program. The control locations of this program are the ideal configurations in $S_\ideal^\#$. With each abstract object occurring in an ideal configuration we associate a counter. The value of each counter denotes the number of concrete objects represented by the associated abstract object. While computing $S_\ideal^\#$, we do some extra book keeping and compute for each transition of $S_\ideal^\#$ a corresponding numerical transition that updates the counters of the counter program. These updates capture how many concrete objects change their representation from one abstract object to another.
A formal definition of such numerical programs can be found in~\cite{StructuralCounterAbs}.

The dynamic package interface $\interface(P)$ of $P$ is a numerical program that is an abstraction of the numerical program associated with $S_\ideal^\#$. The control locations of $\interface(P)$ are the ideal configurations in $S_\ideal^\#$ that correspond to call sites, respectively, return sites to public methods of classes in $P$, in the most general client. A connecting path in $S_\ideal^\#$ for a pair of such call and return sites (along with all covering transitions connecting ideal configurations on the path) corresponds to the abstract execution of a single method call.
We refer to the restriction of the numerical program $S_\ideal^\#$ to such a path and all its covering transitions as a \emph{call program}. Each object mapping of $\interface(P)$ represents a summary of one such call program. Hence, an object mapping of $\interface(P)$ describes, both, how
a method call affects the state of objects in a concrete heap configuration and how many objects are effected.

Note that a call program may contain loops because of loops in the method executed by the call program. The summarization of a call program therefore requires an additional abstract interpretation. The concrete domain of this abstract interpretation is given by transitions of counter programs, i.e., relations between valuations of counters. The concrete fixed point is the transitive closure of the transitions of the call program. The abstract domain provides an appropriate abstraction of numerical transitions. How precisely the package interface captures the possible sequences of method calls depends on the choice of this abstract domain and how convergence of the analysis of the call programs is enforced.
We chose a simple abstract domain of \emph{object mappings} that distinguishes between a constant number, respectively, arbitrary many objects transitioning from an abstract object on the call site of a method to another on the return site. However, other choices are feasible for this abstract domain that provide more or less information than object mappings.

\section{Experiences}\label{sect:references}
We have implemented our system by extending the Picasso tool \cite{Zufferey12:Ideal}.
Picasso uses an ideal abstraction to compute the covering sets of depth-bounded graph rewriting systems.
Our extension of Picasso computes a dynamic package interface from a graph rewriting system that encodes the semantics of the method calls in a package.\footnote{Our tool and the full results of our experiments can be found at: \url{http://pub.ist.ac.at/~zufferey/picasso/dpi/index.html}}

For a graph-rewriting system that represents a package, our tool first computes its covering set. 
Using the elements of the covering set, it then performs unfolding over them with respect to all distinct method calls to derive the object mappings of the DPI of the package. 
The computation of the covering elements and the object mappings are carried out as described in the previous section.

In addition to the \srcf{Viewer} and \srcf{Label} example, described in Section \ref{sect:overview}, 
we have experimented with other examples: a set and iterator package, which we used as our running example in the previous sections, and the JDBC statement and result package. 
In the remainder of this section, we present the DPIs for these packages.

\smallskip
\noindent\emph{Set and Iterator.}
%
We considered a simple implementation of the \srcf{Set} and \srcf{Iterator} classes in which the items in a set are stored in a linked list.
The \srcf{Iterator} class has the usual \srcf{next}, \srcf{has\_next}, and \srcf{remove} methods.
The \srcf{Set} class provides a method \srcf{iterator}, which creates an \srcf{Iterator} object associated with the set, and an \srcf{add} method, which adds a data element to the set.
The interface of the package is meant to avoid raising exceptions of types \srcf{NoSuchElementException} and \srcf{ConcurrentModificationException}.
A \srcf{NoSuchElementException} is raised whenever the \srcf{next} method is called on an iterator of an empty list. 
A \srcf{ConcurrentModificationException} is raised whenever an iterator accesses the set after the set has been modified, 
either through a call to the \srcf{add} method of the set or through a call to the \srcf{remove} method of another iterator. 
An iterator that removes an element can still safely access the set afterwards.
(Similar restrictions apply to other Collection classes that implement \srcf{Iterable}.)

We used the following predicates.
The unary abstraction predicate $\mathit{empty(s)}$ determines whether the size of a \srcf{Set} object $s$ is zero or not.
For \srcf{Iterator} objects, we specified two predicates that rely on the attributes of both the \srcf{Set} and the \srcf{Iterator} classes.  
The predicate $\mathit{sync(i)}$ holds for an \srcf{Iterator} object $i$ that has the same version as its associated \srcf{Set} object.
The predicate $\mathit{mover(i)}$ specifies that the position of an \srcf{Iterator} object $i$ in the list of its associated \srcf{Set} object is less than the size of the set.

\begin{figure}%
\centering
\subfigure[Abstract heap configuration $H_0$ of the set-iterator package using predicates: $empty(s) \equiv s.\mathtt{size} = 0$, $\mathit{synch(i)} \equiv i.\mathtt{iver}=i.\mathtt{iter\_of.sver}$, and $\mathit{mover(i)} \equiv i.\mathtt{pos} <i.\mathtt{iter\_of.size}$.]{\input{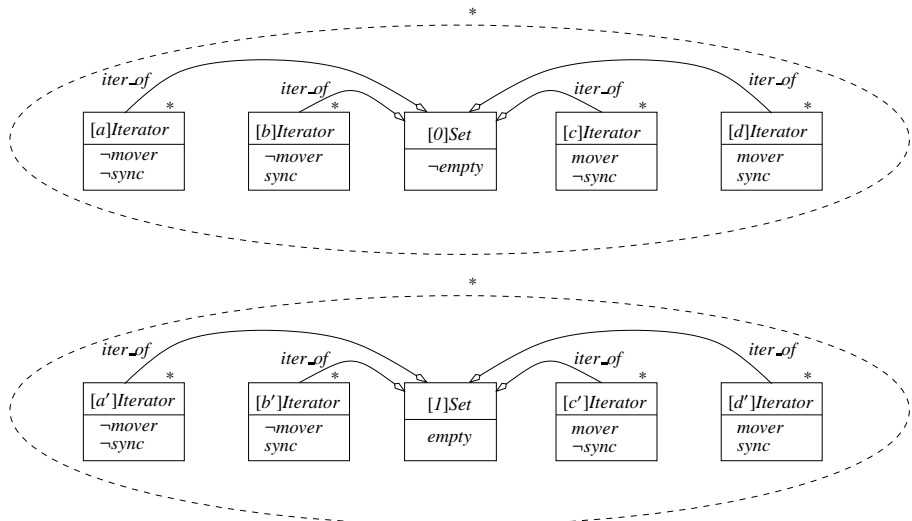}
\label{fig:setitercover1}
}
\subfigure[Object mapping for  $d.\srcf{remove()}$]{\input{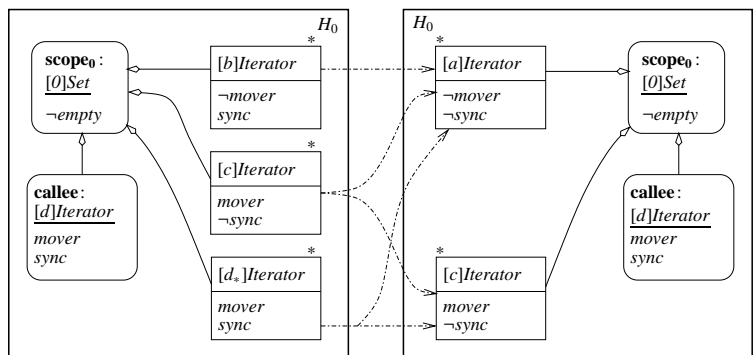}
\label{fig:sit2}
}
\subfigure[Another possible object mapping for $d.\srcf{remove()}$]{\input{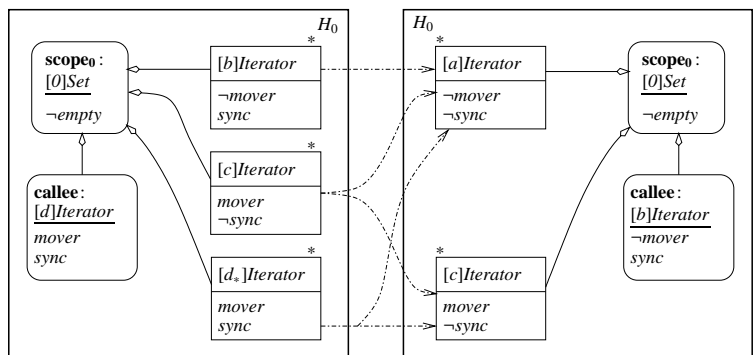}
\label{fig:sit6}
}
\caption{Set-iterator DPI: The abstract heap of the package together with two of its object mappings}
\label{fig:setiter}
\end{figure}

Our algorithm computes the maximal configurations $H_0$, shown in 
Figure~\ref{fig:setitercover1}. 
There are also four error abstract heap configurations, which correspond to different cases in which one
of the two exceptions is raised for an \srcf{Iterator} object.
Figure~\ref{fig:sit2} and \ref{fig:sit6} show the object mappings of two transitions. 
For the sake of clarity, we have omitted the name of the reference attribute $\mathit{iter\_of}$ in the mappings.
While both transitions invoke the \srcf{remove()} method on an \srcf{Iterator} object whose $\mathit{mover}$ and $\mathit{sync}$ predicates 
are true, they have different effects because they capture different concrete heaps represented by the same abstract heap $H_0$. 
The first transition shows the case when the callee object remains a mover, i.e., its \srcf{pos} field does 
not refer to the last element of the list. 
The second transition shows the case when the callee object becomes a non-mover; i.e., before the call to \srcf{remove},
its \srcf{pos} field refers to the last element of the linked list. 
In both transitions, the other \srcf{Iterator} objects that reference the same \srcf{Set} object all become unsynced. Some of these objects remain movers while some of them become non-movers.
In both cases, the callee remains sycned.
There are two other symmetric transitions that capture the cases in which the \srcf{Set} object becomes empty.

\begin{comment}
For the sake of brevity, we have presented the interface for a package which only allows a single \srcf{Set} 
object, however, the interface for the case when there is more than one \srcf{Set} object is similar, 
except that each abstract heap has an extra level of nesting.

\begin{figure}[t]%
\centering
\input{figs/cover1.pstex_t}
\caption{Abstract heap configuration $H_0$ of the set-iterator package using predicates: $empty(s) \equiv s.\mathtt{size} = 0$, $\mathit{synch(i)} \equiv i.\mathtt{iver}=i.\mathtt{iter\_of.sver}$, and $\mathit{mover(i)} \equiv i.\mathtt{pos} <i.\mathtt{iter\_of.size}$.}
\label{fig:setitercover1}
\end{figure}
\end{comment}


\begin{comment}
\begin{figure}[t]%
\centering
%
\subfigure[Object mapping for  $\mathit{\preobj{d}Iterator.\srcf{remove()}}$]{\input{figs/t2.pstex_t}
\label{fig:sit2}
}\quad
\subfigure[Another possible object mapping for $\mathit{\preobj{d}Iterator.\srcf{remove()}}$]{\input{figs/t6.pstex_t}
\label{fig:sit6}
}
\caption{Two of the object mappings of the set-iterator DPI}
\label{fig:setiter}
\end{figure}
\end{comment}

\smallskip
\noindent
\emph{JDBC}
(Java Database Connectivity) is a Java technology that enables access to databases of different types.  
We looked at three classes of JDBC for simple query access to databases: \srcf{Connection}, \srcf{Statement}, and \srcf{ResultSet}.  
A \srcf{Connection} object provides a means to connect to a database.
A \srcf{Statement} object can execute an SQL query statement through a \srcf{Connection} object.  
A \srcf{ResultSet} object stores the result of the execution of a \srcf{Statement} object.  
All objects can be closed explicitly.  
If a \srcf{Statement} object is closed, its corresponding \srcf{ResultSet} object is also implicitly closed.  
Similarly, if a \srcf{Connection} object is closed, its corresponding \srcf{Statement} objects are implicitly closed, and so are the open \srcf{ResultSet} objects 
of these \srcf{Statement} objects. 
Java documentation states: ``By default, only one \srcf{ResultSet} object per \srcf{Statement} object can be open at the same time. 
Therefore, if the reading of one \srcf{ResultSet} object is interleaved with the reading of another, each 
must have been generated by different \srcf{Statement} objects. 
All execution methods in the \srcf{Statement} interface implicitly close a statement's current \srcf{ResultSet} object if an open one 
exists.'' 

Figure \ref{fig:jdbch01} shows the maximal abstract heap $H_0$ computed by our tool.
It represents all safe configurations in which the \srcf{Connection} object is either open or closed. 
Each type of object has a corresponding ``open'' predicate that specifies whether it is open or not.
The node $c$ is of particular interest, as it demonstrates the preciseness of our algorithm: It has the same nesting level as the node $b$, which means that an open \srcf{Statement} object can have at most one open \srcf{ResultSet} object associated with it.
We omit showing abstract heaps capturing erroneous configurations. 
Lastly, Figure \ref{fig:jdbct2} shows the object mapping for the \srcf{close} method call on an open \srcf{Statement} object with an open \srcf{ResultSet} object. 
The mapping takes the \srcf{Statement} object and the open \srcf{ResultSet} object to their corresponding closed objects. 
All other objects remain the same.

\begin{figure}[t]%
\centering
\subfigure[Abstract heap configurations $\mathit{H_0}$ of JDBC package]{\input{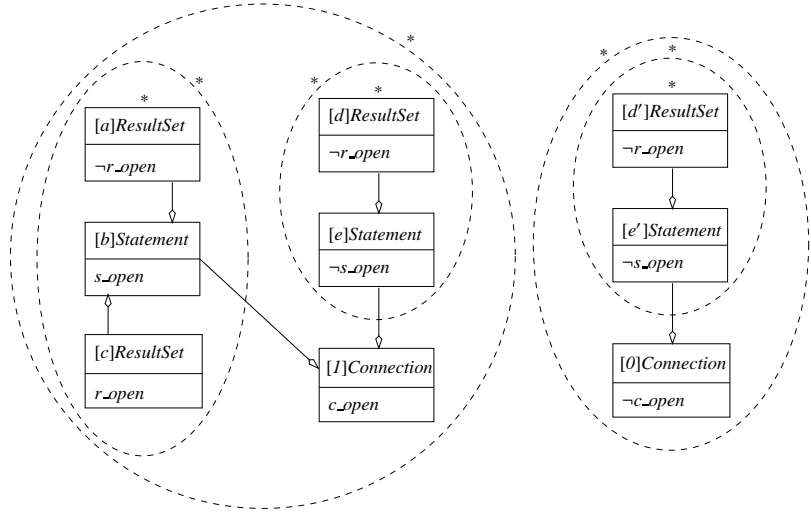}
\label{fig:jdbch01}
}
\subfigure[Object mapping for  $\mathit{b}.\srcf{close()}$]{\input{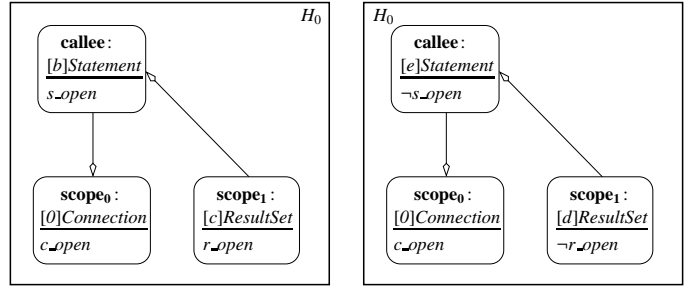}
\label{fig:jdbct2}
}
\caption{JDBC DPI: The abstract heap of JDBC together with one of its object mappings}
\label{fig:jdbcdpi}

\end{figure}



\section{Conclusions}\label{sect:conclu}

We have formalized DPIs for OO packages with inter-object references, developed a novel ideal
abstraction for heaps, and given a sound and terminating algorithm to compute DPIs on the (infinite) abstract domain. 
In contrast to previous techniques for multiple objects based on mixed static-dynamic analysis 
\cite{Nanda05:Deriving,Pradel12:Static}, our algorithm is guaranteed to be sound.
While our algorithm is purely static, an interesting future direction is to effectively combine it with dual, dynamic 
\cite{Ghezzi09:Synthesizing,DBLP:conf/issta/DallmeierKMHZ10,Pradel12:Static} and template-based \cite{DBLP:journals/ase/WasylkowskiZ11} techniques. 

\bibliographystyle{splncs03}

\bibliography{main}

\end{document}